\documentclass{article}

    \PassOptionsToPackage{numbers, compress}{natbib}

\usepackage[preprint]{edited_neurips_2024}

\usepackage[utf8]{inputenc}        
\usepackage[T1]{fontenc}           
\usepackage[hidelinks]{hyperref}   
\usepackage{url}                   
\usepackage{graphicx}
\usepackage{booktabs, multicol}    
\usepackage{enumitem}              
\usepackage{wrapfig}               
\usepackage[table, svgnames]{xcolor}  
\usepackage{amssymb, amsmath, amsfonts, amsthm}
\usepackage[toc,page]{appendix}
\usepackage{titletoc}              
\usepackage{algorithm}             
\usepackage{algpseudocode}         
\usepackage{tikz}                  
\usepackage{soul}                  

\newcommand{\shadedText}[1]{
\noindent\colorbox{teal!5}{
  \parbox{\dimexpr\linewidth-2\fboxsep}{
      #1
    }
  }
}

\definecolor{YaleBlue}{rgb}{0.059,0.302,0.573}
\definecolor{forestgreen}{rgb}{0.133,0.549,0.133}
\definecolor{crimson}{rgb}{0.863,0.078,0.235}

\begin{document}

%
\title{ImageFlowNet: Forecasting Multiscale Image-Level Trajectories of Disease Progression with Irregularly-Sampled Longitudinal Medical Images}

\author{%
\textbf{Chen Liu}$^{1 *}$ \quad
\textbf{Ke Xu}$^{1 *}$ \quad
\textbf{Liangbo L. Shen}$^{2}$ \quad
\textbf{Guillaume Huguet}$^{3,4}$ \quad
\textbf{Zilong Wang}$^{3,5}$\\
\textbf{Alexander Tong}$^{3,4 \S}$ \quad
\textbf{Danilo Bzdok}$^{3,5 \S}$ \quad
\textbf{Jay Stewart}$^{2 \S}$ \quad
\textbf{Jay C. Wang}$^{1,2,6 \S}$\\
\textbf{Lucian V. Del Priore}$^{1 \S}$ \quad
\textbf{Smita Krishnaswamy}$^{1 \S}$ \vspace{6pt}\\
{\small $^1$Yale University \quad $^2$University of California, San Francisco \quad $^{3}$Mila - Quebec AI Institute}\\
{\small $^{4}$Universit\'e de Montr\'eal \quad $^{5}$McGill University \quad $^{6}$Northern California Retina Vitreous Associates \vspace{6pt}}\\
{\small *~These authors are joint first authors: \texttt{\{chen.liu.cl2482, k.xu\}@yale.edu}. \quad $\S$~Senior authors.}\\
{\small Please direct correspondence to: \url{smita.krishnaswamy@yale.edu} or \url{lucian.delpriore@yale.edu}.}
}

\makeatletter
\let\inserttitle\@title
\makeatother

\definecolor{wine}{HTML}{830E0D}
\definecolor{color_ode}{RGB}{186, 108, 73}
\definecolor{color_unet}{HTML}{1365C0}
\definecolor{nicered}{HTML}{B22222}
\definecolor{niceblue}{HTML}{0000FF}
\newcommand{\cc}[0]{\cellcolor{color_unet!10}}

\newcommand{\fancynumber}[1]{%
  \tikz[baseline=(char.base)]{
    \node[shape=circle,draw=black,fill=wine,inner sep=1pt](char){\color{white}#1};
  }%
}

\newtheorem{theorem}{Theorem}[section]
\newtheorem{lemma}[theorem]{Lemma}
\newtheorem{prop}[theorem]{Proposition}
\newtheorem{cor}{Corollary}
\newtheorem{definition}{Definition}[section]
\newtheorem{conj}{Conjecture}[section]
\newtheorem{rem}{Remark}

\maketitle
\begin{abstract}
Advances in medical imaging technologies have enabled the collection of longitudinal images, which involve repeated scanning of the same patients over time, to monitor disease progression. However, predictive modeling of such data remains challenging due to high dimensionality, irregular sampling, and data sparsity. To address these issues, we propose ImageFlowNet, a novel model designed to forecast disease trajectories from initial images while preserving spatial details. ImageFlowNet first learns multiscale joint representation spaces across patients and time points, then optimizes deterministic or stochastic flow fields within these spaces using a position-parameterized neural ODE/SDE framework. The model leverages a UNet architecture to create robust multiscale representations and mitigates data scarcity by combining knowledge from all patients. We provide theoretical insights that support our formulation of ODEs, and motivate our regularizations involving high-level visual features, latent space organization, and trajectory smoothness. We validate ImageFlowNet on three longitudinal medical image datasets depicting progression in geographic atrophy, multiple sclerosis, and glioblastoma, demonstrating its ability to effectively forecast disease progression and outperform existing methods. Our contributions include the development of ImageFlowNet, its theoretical underpinnings, and empirical validation on real-world datasets. The official implementation is available at \url{https://github.com/KrishnaswamyLab/ImageFlowNet}.

\end{abstract}

\section{Introduction}

Advances in medical imaging technologies such as X-ray, computed tomography~(CT), optical coherence tomography~(OCT), and magnetic resonance imaging~(MRI) combined with improved storage capacity and practices have enabled the collection of longitudinal medical images that track disease progression~\cite{ImagingInHealthcare, longitudinal_image1, longitudinal_image2}. However, predictive modeling using such data is challenging due to the high dimensionality of images, irregular time intervals between samples, and the sparsity of data in most patients (see Appendix~\ref{supp:additional_background} for more background). These challenges often lead to methods that undermine the spatial-temporal nature of the data and instead treat them as time series of hand-crafted features~\cite{EBM1, EBM2, DEBM, lu2024cats, ruan2024comprehensive, lyu2023multimodal, guo2024weits, sun2023manifold, dong2023integrated, DCD, Disease_prog_xgboost, Disease_prog_polynomial_network, Disease_prog_lstm_covid19, Disease_prog_transformer1, Disease_prog_lstm_AD}, losing rich spatial information within the images (see Fig.~\ref{fig:advantage}, top panel).
\begin{wrapfigure}{l}{0.65\textwidth}
    \centering
    \includegraphics[width=0.65\textwidth]{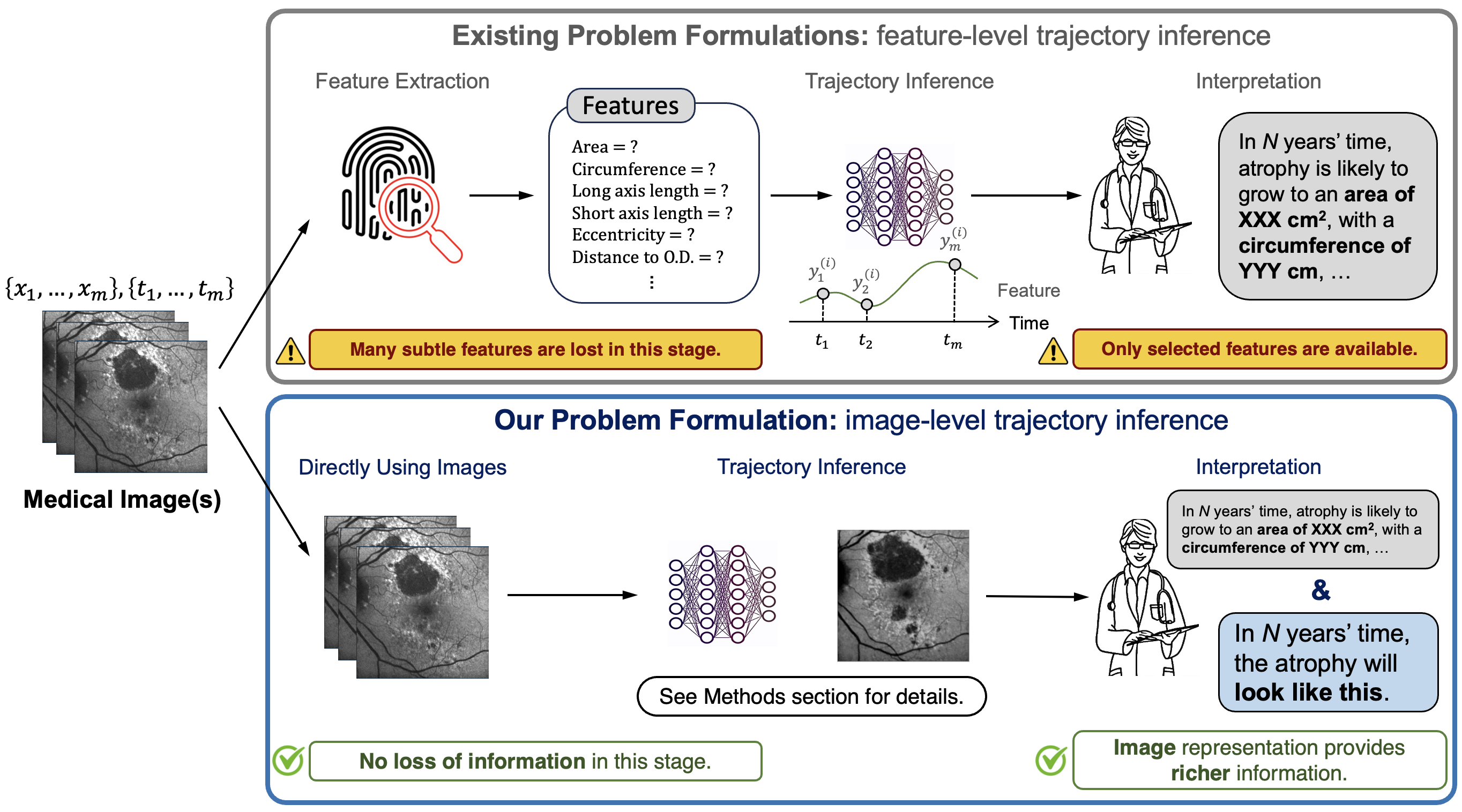}
    \caption{Advantages of image-level trajectory inferece.}
    \label{fig:advantage}
    \vskip -4pt
\end{wrapfigure}
To address all these issues, we introduce ImageFlowNet, a model designed to forecast disease progression from initial images while addressing the aforementioned challenges and preserving spatial detail~(see Fig.~\ref{fig:advantage}, bottom panel). ImageFlowNet learns multiscale joint representation spaces from all patients' images and optimizes deterministic or stochastic flow fields over these spaces using a modified position-parameterized neural ODE/SDE framework.

ImageFlowNet learns multiscale representations with a UNet~\cite{UNet} backbone (implementation adopted from~\cite{DDPM}) while enforcing robustness to variations such as scaling, rotation, and contrast through extensive augmentation. Next, a vector field representing flows in each representation space is learned with a position-parameterized neural ODE framework. Unlike the standard parameterization by time $f_\theta(z_t, t)$, we parameterize the derivative by each vector's position in the embedding space $f_\theta(z_t)$, which ensures that the space is shared across all patients at all times. Due to the nature of such joint embedding spaces, this approach mitigates patient-level data scarcity. We also present a stocastic alternative that generates non-deterministic trajectories.

In addition to presenting these networks, we theoretically establish the equivalent expressive power of the ODE and demonstrate the connections of ImageFlowNet to dynamic optimal transport. We then present empirical results on one longitudinal retinal image dataset with geographic atrophy and two longitudinal brain image datasets with multiple sclerosis and glioblastoma. Our main contributions are as follows.

\vskip -8pt
\begin{enumerate}[leftmargin=22pt]
    \item Proposing ImageFlowNet to forecast trajectories of disease progression in the image domain.
    \item Learning multiscale joint patient representation spaces that integrate the knowledge from all observed trajectories and remedy the data scarcity issue of any single patient.
    \item Designing a multiscale position-parameterized ODE/SDE and providing theoretical rationales.
    \item Showcasing results on three medical image datasets with sparse longitudinal progression data. 
\end{enumerate}
\vskip -12pt

\section{Preliminaries and Background}
\label{sec:prelim}

\paragraph{Problem Formulation and Notation}
We consider a set of $N$ longitudinal image series $\{X^{(m)}\}_{m=1}^N$, where the $m$-th series $X^{(m)}$ contains $n_m \geq 2$ images and all images have the same dimension $\mathbb{R}^{H \times W \times C}$. These images $X^{(m)} = \{ x^{(m)}_1, x^{(m)}_2, ..., x^{(m)}_{n_m} \}$ are acquired at time points $T^{(m)} = \{ t^{(m)}_1, t^{(m)}_2, ..., t^{(m)}_{n_m} \}$ where $t^{(m)}_1 < t^{(m)}_2 < ... < t^{(m)}_{n_m} \in \mathbb{R}$. Note that we do not assume any specific sampling schedule, such as uniform sampling over time. In addition, the time points for different series are not necessarily the same. This represents a very common scenario of a medical record containing $N$ patients with multiple visits, where visit schedules can be irregular over time and heterogeneous among patients. For simplicity, we will omit the script $m$ when considering the same series. Our task is to predict any image $x_k$ with $1 < k \leq n$ given any subset of the earlier images $\tilde{X}_{<k}$ and their corresponding time points $\tilde{T}_{<k}$, where $\varnothing \subset \tilde{X}_{<k} \subseteq \{ x_1, x_2, ..., x_{k-1} \}$.


\paragraph{Neural Ordinary Differential Equations (Neural ODEs)}
Neural ODEs~\cite{NeuralODE} model the evolution of a variable over time by considering the ODE in Eqn~\eqref{eqn:neural_ode_diffeq}, where $f_\theta$ is parameterized by a neural network. Since the gradient field $f_\theta$ is defined at every time point, future states can be modeled deterministically from an earlier state by integration, as shown in Eqn~\eqref{eqn:neural_ode_infer}.

\begin{subequations}
\noindent\begin{minipage}[t]{0.4\textwidth}
\begin{equation}
\frac{\mathrm{d} y(\tau)}{\mathrm{d}\tau} = f_\theta(y(\tau), \tau)
\label{eqn:neural_ode_diffeq}
\end{equation} 
\end{minipage}
\begin{minipage}[t]{0.58\textwidth}
\begin{equation}
y(t_1) = y(t_0) + \int_{t_0}^{t_1} f_\theta(y(\tau), \tau) \mathrm{d}\tau
\label{eqn:neural_ode_infer}
\end{equation}
\end{minipage}
\end{subequations}

In practice, the integration step is performed by an ODE solver: $y(t_0) + \int_{t_0}^{t_1} f_\theta(y(\tau), \tau) \mathrm{d}\tau = \mathrm{ODESolve}(y(t_0), f_\theta, t_0, t_1)$. The gradient field $f_\theta$ can be optimized by any loss function $L(\cdot)$ that takes the result from the solver as input, in the form of $L(\mathrm{ODESolve}(y(t_0), f_\theta, t_0, t_1), \mathrm{*args})$.

\paragraph{Neural Stochastic Differential Equations~(Neural SDEs)}



Neural SDEs~\cite{NeuralSDE} inject stochasticity into deterministic neural ODEs by additionally considering Brownian motion $\{ W_t \}_{t \geq 0}$ in the equation (Eqn~\eqref{eqn:neural_sde_diffeq}). $f_\theta$ and $\sigma_\phi$ respectively model the drift and diffusion components.

\vskip -12pt
\begin{align}
\label{eqn:neural_sde_diffeq}
\mathrm{d} y(\tau) &= f_\theta(y(\tau), \tau) \mathrm{d}\tau + \sigma_\phi(y(\tau), \tau) \mathrm{d}W_\tau
\end{align}

\paragraph{UNet}
UNet~\cite{UNet} is a convolutional neural network architecture originally designed for biomedical image segmentation, but has later been found competent in many tasks such as image-to-image translation~\cite{pix2pix}, style transfer~\cite{style_transfer_UNet}, and image generation as in diffusion models~\cite{DDPM}. It has a distinctive U-shaped structure with a contraction path that extracts multiscale features and an expansion path that recovers spatial resolution, along with skip connections that achieve residual learning~\cite{ResNet}.

\section{ImageFlowNet}

\begin{figure*}[!bth]
    \centering
    \includegraphics[width=0.9\textwidth]{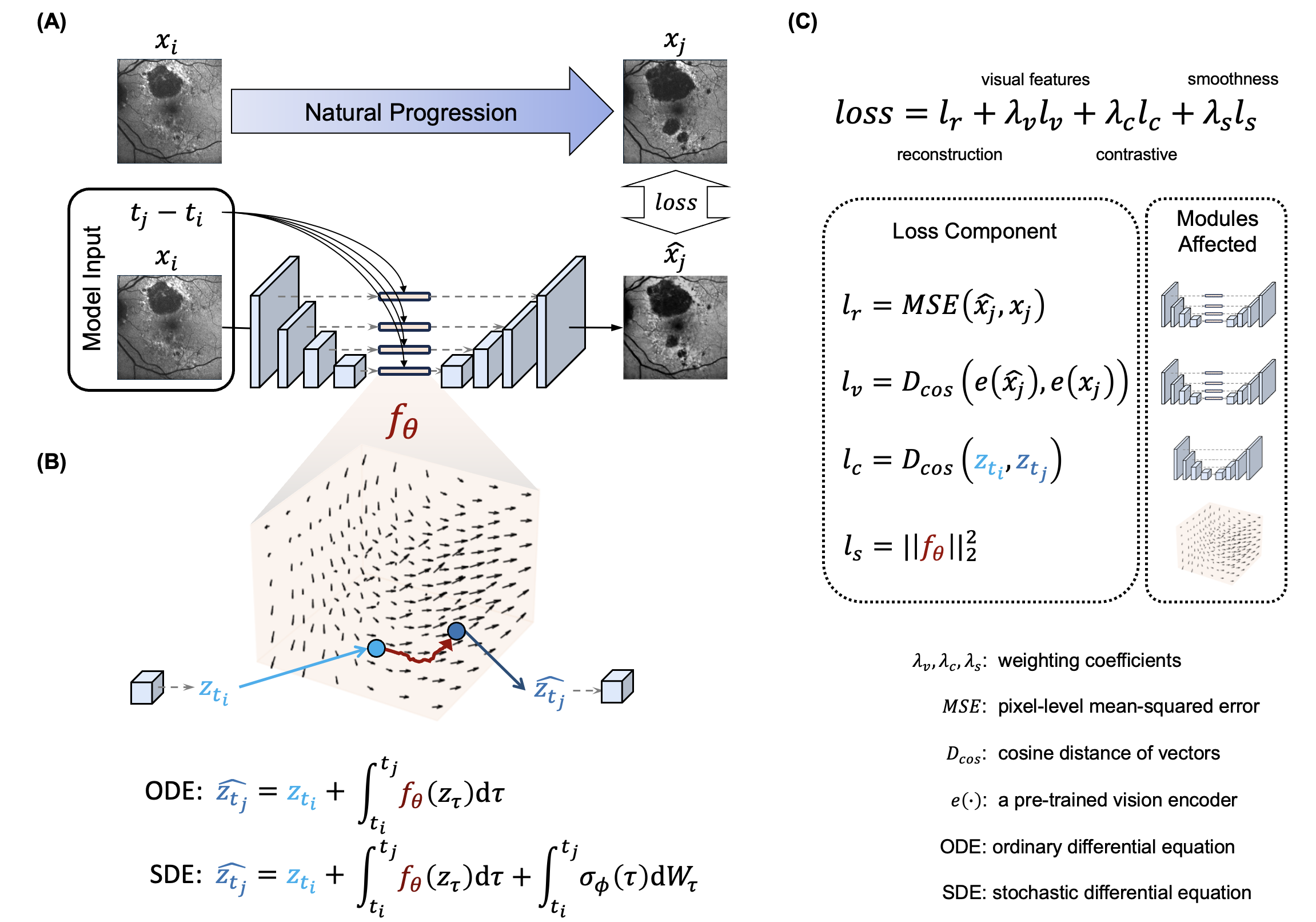}
    \caption{Overview of the proposed ImageFlowNet. (A) The model uses an earlier image $x_i$ at time $t_i$ as well as the change in time $t_j - t_i$ to forecast the future image $x_j$ at time $t_j$. (B) For each hidden layer, a separate flow field $f_\theta$ is used to model the joint patient embedding space. Trajectory inference can be performed by integration along this flow field. It should be noted that the change in time $t_j - t_i$ is sufficient for integration in practice, while the exact time values $t_i$ and $t_j$ are included in the integral merely for mathematical clarity. (C) The learning objective has four components. The loss function and modules affected by each component are illustrated. }
    \label{fig:architecture}
\end{figure*}

ImageFlowNet models the spatial-temporal dynamics of longitudinal images by first establishing a joint patient representation space~(\textcolor{YaleBlue}{Section~\ref{sec:learn_joint_patient_rep}}), and then flowing the representations of earlier time points to later time points~(\textcolor{YaleBlue}{Section~\ref{sec:learn_flow_field}}). This approach incorporates all images at all times during training, and thus addresses the issue of data scarcity at the individual patient level without impairing the inference capabilities. From an engineering perspective, ImageFlowNet extracts latent representations at various resolutions and reassembles them to an image (i.e., the \textit{spatial} aspect) after evolving these latent representations over time (i.e., the \textit{temporal} aspect) along a learned flow field.

\subsection{Learning Multiscale Spaces of Joint Patient Representations}
\label{sec:learn_joint_patient_rep}

As illustrated in Fig.~\ref{fig:architecture}~(A), we first learn joint embedding spaces for training samples in the hidden layers within the contraction path of a UNet. Multiscale representations are extracted from an input image $x_i$ acquired at time $t_i$ to produce $B$ representations, one for each hidden layer, at $R$ resolutions, which we denote $\{ z_{t_i}^{(b)} \}_{b=1}^B$, with $B \geq R$.

During training, images are augmented with transformations that may naturally occur during acquisition, including reflection, rotation, shifting, rescaling, random brightness and contrast, and additive noise. Augmented versions effectively enlarge the sample variety and better populate the joint embedding spaces. This increases the chance that, during inference, a new image is embedded close to images seen in the training set and can leverage the learned dynamics around that local cohort.

\subsection{Learning Multiscale Flow Fields on Joint Patient Representations}
\label{sec:learn_flow_field}

As depicted in Fig.~\ref{fig:architecture}~(B), at each hidden layer that spans different granularities of the image, a flow field is learned to evolve the joint patient representations at that scale. The flow field $f_\theta^{(b)}$ parameterizes the flow gradient as described in Eqn~\eqref{eqn:imageflownet_diffeq}, so that, given an initial position and a time duration, the trajectory can be computed through integration. $f_\theta^{(b)}$ is implemented as a 2-layer convolutional neural network whose input and output dimensions match the dimension of $z_{t_i}^{(b)}$. The latent representation $z_{t_j}^{(b)}$ corresponding to the future time $t_j$ can be inferred using Eqn~\eqref{eqn:imageflownet_infer}, using Eqn~\eqref{eqn:imageflownet_expansion_path} or similar variants. Finally, these multiscale representations meet at the expansion path to compose an output image. From now on, we will omit the superscript $\cdot^{(b)}$ if the specific hidden layer is not emphasized.

\vskip -16pt
\begin{subequations}
\noindent\begin{minipage}[t]{0.43\textwidth}
\begin{equation}
\frac{\mathrm{d} z_{\tau}^{(b)}}{\mathrm{d}\tau} = f_\theta^{(b)}(z_{\tau}^{(b)})  \quad \text{ for } b \in [1, B]
\label{eqn:imageflownet_diffeq}
\end{equation} 
\end{minipage}
\begin{minipage}[t]{0.57\textwidth}
\begin{equation}
z_{t_j}^{(b)} = z_{t_i}^{(b)} + \int_{t_i}^{t_j} f_\theta^{(b)} (z_\tau^{(b)}) \mathrm{d}\tau \quad \text{ for } b \in [1, B]
\label{eqn:imageflownet_infer}
\end{equation}
\end{minipage}
\end{subequations}

\vskip -10pt
\begin{align}
\begin{split}
    \widehat{x_j} &= \text{ResBlock}(\text{Concat}(\tilde{z}_{t_j}^{(2)}, z_{t_j}^{(1)})) \text{, where}\\
    \tilde{z}_{t_j}^{(b)} &= \text{Upsample}(\text{ResBlock}(\text{Concat}(\tilde{z}_{t_j}^{(b+1)}, z_{t_j}^{(b)}))) \text{ for } b \in [2, B-1] \text{, with } \tilde{z}_{t_j}^{(B)} = z_{t_j}^{(B)}
\end{split}
\label{eqn:imageflownet_expansion_path}
\end{align}
\vskip -20pt

\subsection{Training and Inference}
\paragraph{Training Objectives}
As in other neural differential frameworks, we use an ODE solver to compute the integral. Any loss function on the inferred latent representation $z_{t_j}$ can be backpropagated through the ODE solver. Since the inferred image $\widehat{x_{j}}$ only depends on the expansion path of ImageFlowNet and the inferred latent representations $\{ z_{t_i}^{(b)} \}_{b=1}^B$, the same principle applies to the loss functions on $\widehat{x_{j}}$ as well.

As shown in Eqn~\eqref{eqn:loss_total}, our loss function contains four components, which we will explain below. During training, the learnable parameters in the entire ImageFlowNet are penalized by the first
two components, while contrastive regularization only affects the UNet backbone and smoothness
regularization only affects the flow field, as described in Fig.~\ref{fig:architecture}~(C).

\vskip -10pt
\begingroup\makeatletter\def\f@size{9}\check@mathfonts
\def\maketag@@@#1{\hbox{\m@th\large\normalfont#1}}%
\begin{align}
\begin{split}
    loss =& \quad \overset{\text{\fancynumber{1} $l_r$: reconstruction}}{\overbrace{
    \frac{1}{H W C} \sum_{h \in H} \sum_{w \in W} \sum_{c \in C} || \widehat{x_j}[h, w, c] - x_j[h, w, c] ||_2^2
    }} +
    \overset{\text{\fancynumber{2} $l_v$: visual feature}}{\overbrace{
    \lambda_v \left(- \frac{e(\widehat{x_j})^\top e(x_j)}{||e(\widehat{x_j})||_2 ||e(x_j)||_2} \right)
    }} \\
    +& \hspace{3pt} \overset{\text{\fancynumber{3} \color{color_unet} $l_c$: contrastive learning (SimSiam)}}{\overbrace{
    \lambda_c \left(-\frac{p_d(p_j(z_{t_i}))^\top p_j(z_{t_j}) }{2||p_d(p_j(z_{t_i}))||_2 ||p_j(z_{t_j})||_2} -\frac{p_d(p_j(z_{t_j}))^\top p_j(z_{t_i}) }{2||p_d(p_j(z_{t_j}))||_2 ||p_j(z_{t_i})||_2} \right)
    }} + 
    \overset{\text{\fancynumber{4} \color{color_ode} $l_s$: trajectory smoothness}}{\overbrace{
    \lambda_s || f_\theta ||_2^2
    }}
\end{split}
\label{eqn:loss_total}
\end{align}
\endgroup


\noindent\fancynumber{1}~\textit{Image reconstruction} is achieved by MSE, attending to low-level features on the pixel level.

\noindent\fancynumber{2}~\textit{Visual feature regularization} guides the network to produce images that resemble the ground truth on high-level features judged by a ConvNeXT~\cite{ConvNeXT} encoder pretrained on ImageNet~\cite{ImageNet}.

\noindent\fancynumber{3}~\textit{Contrastive learning regularization} organizes a well-structured ImageFlowNet latent space, by encouraging proximity of representations from images within the same longitudinal series, following the SimSiam formulation~\cite{SimSiam}.

\noindent\fancynumber{4}~\textit{Trajectory smoothness regularization} leverages a theorem in convex optimization~(Lemma 2.2 in~\cite{Book_OptDataAnalysis}) to enforce smoothness of trajectories by regularizing the norm of the field. \textit{Notably, this achieves Lipschitz continuity, satisfying a crucial assumption for our theoretical results.}

\noindent\textbf{Full Trajectories Are Used for Training}
Although training is performed over pairs of observations, since all samples within the longitudinal series are embedded in the same flow field $f_\theta$ and trajectories between all pairs are used for optimization, the model is indeed trained by the full trajectories.

\noindent\textbf{Inferring Trajectories}
During inference, for a new patient with one or more previous observations, we can use any observation to serve as the starting point and obtain the future prediction using Eqn~\eqref{eqn:imageflownet_infer}. This approach has a low demand on access to patient history. In case we really want to take advantage of the entire patient history, we can perform test-time optimization to fine-tune $f_\theta$ using the measured time series and infer the patient trajectory afterwards (see Section~\ref{sec:test_time_optimization}).

\paragraph{ImageFlowNet$_\textrm{SDE}$ as a Stochastic Variant}

We formulate a Langevin-type SDE as an alternative to our ODE, given by $\mathrm{d} z_{\tau}^{(b)} = f_\theta^{(b)}(z_{\tau}^{(b)}) \mathrm{d}\tau + \sigma_\phi^{(b)}(\tau) \mathrm{d}W_\tau \hspace{2pt} \text{ for } b \in [1, B]$. This is motivated by stochasticity in patient trajectories and the need to model alternative outcomes from the same starting point, which is not possible in a deterministic ODE. In this construction, we decompose the dynamics into a deterministic drift function $f_\theta$ and a stochastic diffusion term $\sigma_\phi$. This SDE is guaranteed to have a unique strong solution under mild assumptions~\cite{StableNeuralSDE}.


An important characteristic unique to SDEs is their ability to model alternative trajectories due to the stochastic diffusion term. This feature enables SDEs to estimate uncertainty by generating multiple trajectories to infer the distribution of potential outcomes. By analyzing the variability and frequency of these outcomes, we can quantify the likelihood of various progression trends.

\section{Theoretical Results}

\paragraph{A Position-Parameterized ODE}
Unlike the original version, our ODE is parameterized by time \textit{indirectly}: $f_\theta(z_t)$ instead of $f_\theta(z_t, t)$. This altered formulation learns a more compact vector field which alleviates the data scarcity issue, and allows the space to model disease states regardless of the time from disease onset. Importantly, we demonstrate that this formulation is \ul{theoretically equivalent in expressiveness} to the original formulation (Proposition~\ref{prop:expressiveness_equivalence}) and achieves \ul{better empirical performance} (Table~\ref{tab:gradient_field} in the Empirical Results section). The proof is shown in Appendix~\ref{supp:proofs}.

\shadedText{
\begin{prop}
\label{prop:expressiveness_equivalence}
Let $f_{\theta}$ be a continuous function that satisfies the Lipschitz continuity and linear growth conditions. Also, let the initial state $y(t_0) = y_0$ satisfy the finite second moment requirement $(\mathbb{E}[\|y(t_0)\|^2] < \infty)$. Suppose $z(t_0)$ is the latent representation learned by ImageFlowNet in the initial state corresponding to $t_0$. Then, our neural ODEs $($Eqn~\eqref{eqn:imageflownet_diffeq}$)$ are at least as expressive as the original neural ODEs $($Eqn~\eqref{eqn:neural_ode_diffeq}$)$, and their solutions capture the same dynamics.
\end{prop}
\vskip -4pt
}
\vskip -16pt

\paragraph{Connections to Dynamic Optimal Transport} 
Dynamic Optimal Transport~\cite{BenamouAndBrenier} aims to find the optimal transport plan to achieve the minimal transport cost between the original and target distribution. Here, we show that ImageFlowNet falls into the framework of dynamic optimal transport of images with ground distance based on the UNet representation. The proof is shown in Appendix~\ref{supp:proofs}.

\shadedText{
\begin{prop}
\label{prop: equivalence_DOT}
If we consider an image as a distribution over a 2D grid, ImageFlowNet is equivalently solving a dynamic optimal transport problem, as it meets three essential criteria: (1)~matching the density, (2) smoothing the dynamics, and (3) minimizing the transport cost, where the ground distance is the Euclidean distance in the latent joint embedding space.
\end{prop}
}
\vskip -16pt

\section{Empirical Results}

\subsection{Preprocessing: Aligning Longitudinal Images with Image Registration}
Longitudinal image datasets often face the problem of spatial misalignment among images acquired at different time points. This phenomenon is nearly inevitable, as minor adjustments in position or angle can disrupt the exact alignment between pixels. To address this problem, we spatially align all images in a longitudinal series during the preprocessing stage using keypoint detection~\cite{SuperRetina} along with a perspective transform for retinal images, and affine registration~\cite{ANTS} for brain scans. More details are described in Appendix~\ref{supp:image_registration}.

\subsection{Baseline Methods}

\textbf{Image extrapolation methods} \hspace{2pt} is the most straightforward method for inferring future images. We included linear extrapolation~\cite{LinearInterp} and cubic spline extrapolation~\cite{CubicSpline} for comparison. 

\textbf{Time-conditional UNet~(T-UNet)}  integrates time by adding a time-embedding tensor to representations throughout UNet and is a key component of diffusion models~\cite{DDPM, DDIM}. The sinusoidal waveform is commonly used for time embedding, similar to the position encoding in transformers~\cite{Transformer}.

\textbf{Time-aware diffusion model~(T-Diffusion)} \hspace{2pt} is a modification to existing diffusion models by changing the diffusion time schedule. We introduce time awareness by considering the diffusion steps as equally spaced time intervals and dynamically adjusting the number of diffusion steps to match the time gap. Furthermore, our implementation is based on a specific diffusion model called image-to-image Schr\"odinger bridge (I2SB)~\cite{I2SB}, which directly maps from the input image to the output image without a noising and denoising process as in many others~\cite{DDPM, DDIM, BBDM}. This allows it to produce high-quality images at any arbitrary diffusion time step, which is critical to our use case.

\subsection{Datasets}

\paragraph{Retinal Images}
We used longitudinal retinal images from the METforMIN study~\cite{METforMIN, CUTS} that monitored patients across 12 clinical centers with geographic atrophy, an advanced form of age-related macular degeneration that is slowly progressive and can lead to loss of vision. The dataset contains fundus autofluroscence images of 132 eyes over 2-5 visits at irregular intervals for a duration of up to 24 months. 

\paragraph{Brain Multiple Sclerosis Images}
We used longitudinal FLAIR-weighted MRI scans from the LMSLS dataset~\cite{carass2017longitudinal} monitoring patients with Multiple Sclerosis~(MS) over an average of 4.4 time points for approximately 5 years. We obtained 79 longitudinal image series from this dataset.

\paragraph{Brain Glioblastoma Images}
We used longitudinal contrastive T1-weighted MRI scans from the LUMIERE dataset~\cite{LUMIERE_dataset} that tracked 91 glioblastoma~(GBM) patients who underwent a pre-operative scan and repeated post-operative scans for up to 5 years. We obtained a set of 795 longitudinal image series each with 2-18 time points. Only post-operative images are kept in each series to model the natural change of tissues after surgery.

For all datasets, we took caution to avoid data leakage: data were partitioned on the level of longitudinal series into train/validation/test sets, and images from the same patient would only go to the same set. Disease regions were delineated by standalone image segmentation networks, one for each dataset. We quantified image similarity using peak signal-to-noise ratio~(PSNR) and structural similarity index~(SSIM), residual map similarity using mean absolute error~(MAE) and mean squared error~(MSE), and disease region overlap using Dice-S{\o}rensen coefficient~(DSC) and Hausdorff distance~(HD). Disease regions were delineated by 3 standalone image segmentation networks, one for each dataset. Mean values and standard deviations are reported from 3 independent runs with different random seeds. Training and implementation details, as well as information about the metrics, can be found in Appendices~\ref{supp:implementation_details} and \ref{supp:metrics}.

\subsection{ImageFlowNet Forecasts Images and Preserves Visual Traits of Disease Progression}

As shown in Figure~\ref{fig:comparison}, the classic extrapolation methods and the time-conditional UNet often struggle to capture the changes in the disease-affected regions, implying that these methods are not quite capable of modeling the complex time evolution of the underlying disease processes. This is especially obvious when inspecting the residual maps. Extrapolation methods work significantly better on the glioblastoma dataset, where more historical data is available during inference. The diffusion model can represent atrophy growth in the retinal image, but not in the other cases. On the other hand, our proposed ImageFlowNet variants better represent disease-related changes.

As summarized in Table~\ref{tab:main}, our proposed methods achieve improved quantitative results, as demonstrated by higher image similarity, smaller residual maps, and better prediction of atrophy. The final ranking indicates that our SDE formulation using visual feature, contrastive learning, and trajectory smoothness regularizations is the best, while our ODE formulation with the three regularizations comes second, followed by the same two models penalized by the image reconstruction loss only.

\begin{figure*}[!tb]
    \centering
    \includegraphics[width=0.95\textwidth]{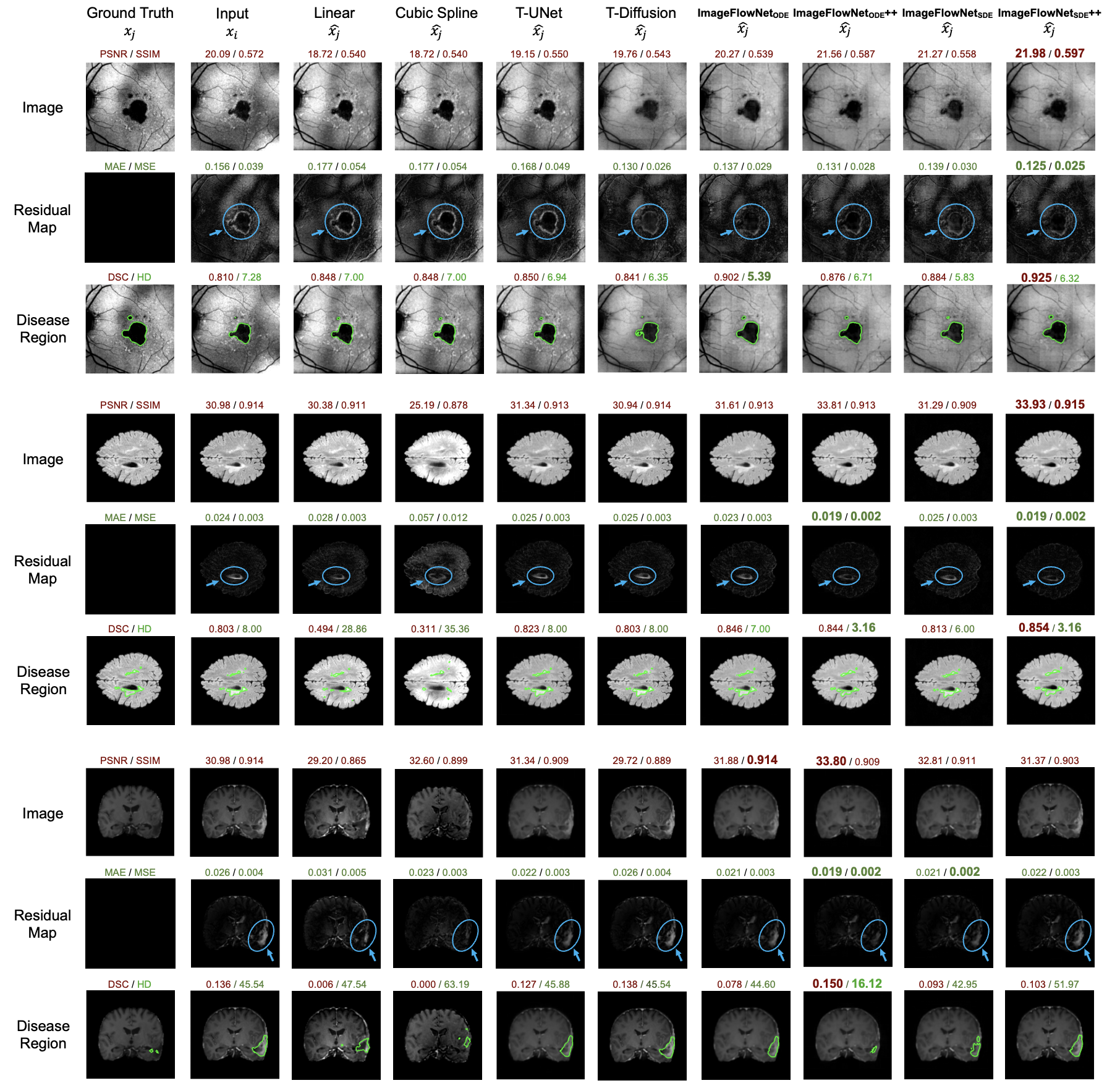}
    \caption{Qualitative comparison of image forecasting results on the retinal geographic atrophy, multiple sclerosis and glioblastoma datasets. ``++'' means using the 3 regularizations in Eqn~\eqref{eqn:loss_total}.}
    \label{fig:comparison}
    \vskip -6pt
\end{figure*}

\begin{table*}[!tb]
    \caption{Image forecasting performance: $\mathrm{metric}(x_j, \widehat{x_j}). \quad \widehat{x_j} = \mathcal{F}(x_i, t_i, t_j), \forall i < j$. \\$^{\dag}$Extrapolation methods use the entire history. ``++'' means using the 3 regularizations in Eqn~\eqref{eqn:loss_total}.}
    \small
    \centering
    \scalebox{0.65}{
    \begin{tabular*}{1.51\linewidth}{lrllllllll}
    \toprule
    Dataset & Metric & Linear$^{\dag}$ & Cubic Spline$^{\dag}$ & T-UNet & T-Diffusion &
    {\scriptsize ImageFlowNet$_{\textrm{ODE}}$} &
    {\scriptsize ImageFlowNet$_{\textrm{ODE}}$++} &
    {\scriptsize ImageFlowNet$_{\textrm{SDE}}$} &
    {\scriptsize ImageFlowNet$_{\textrm{SDE}}$++} \\
    && \cite{LinearInterp} & \cite{CubicSpline} & \cite{GuidedDiffusion} & \cite{I2SB} & \textbf{(ours)} & \textbf{(ours)} & \textbf{(ours)} & \textbf{(ours)} \\
    \toprule
    \textbf{Retinal} &
    PSNR~$\uparrow$ &
    20.22{\scriptsize \color{gray}{$\pm$ 0.00}} &
    19.79{\scriptsize \color{gray}{$\pm$ 0.00}} &
    22.06{\scriptsize \color{gray}{$\pm$ 0.33}} &
    22.29{\scriptsize \color{gray}{$\pm$ 0.33}} & \cc
    22.63{\scriptsize \color{gray}{$\pm$ 0.26}} & \cc
    \underline{22.74}{\scriptsize \color{gray}{$\pm$ 0.25}} & \cc
    22.32{\scriptsize \color{gray}{$\pm$ 0.29}} & \cc
    \textbf{22.89}{\scriptsize \color{gray}{$\pm$ 0.28}} \\
    \textbf{Images} &
    SSIM~$\uparrow$ &
    0.535{\scriptsize \color{gray}{$\pm$ 0.000}} &
    0.505{\scriptsize \color{gray}{$\pm$ 0.000}} &
    0.635{\scriptsize \color{gray}{$\pm$ 0.015}} &
    0.624{\scriptsize \color{gray}{$\pm$ 0.016}} & \cc
    0.646{\scriptsize \color{gray}{$\pm$ 0.012}} & \cc
    0.647{\scriptsize \color{gray}{$\pm$ 0.013}} & \cc
    \textbf{0.651}{\scriptsize \color{gray}{$\pm$ 0.015}} & \cc
    \textbf{0.651}{\scriptsize \color{gray}{$\pm$ 0.012}} \\
    \textit{all} & MAE~$\downarrow$ &
    0.163{\scriptsize \color{gray}{$\pm$ 0.000}} &
    0.177{\scriptsize \color{gray}{$\pm$ 0.000}} &
    0.126{\scriptsize \color{gray}{$\pm$ 0.005}} &
    0.122{\scriptsize \color{gray}{$\pm$ 0.004}} & \cc
    0.119{\scriptsize \color{gray}{$\pm$ 0.004}} & \cc
    \underline{0.118}{\scriptsize \color{gray}{$\pm$ 0.004}} & \cc
    0.124{\scriptsize \color{gray}{$\pm$ 0.005}} & \cc
    \textbf{0.115}{\scriptsize \color{gray}{$\pm$ 0.004}} \\
    \textit{cases} &
    MSE~$\downarrow$ &
    0.050{\scriptsize \color{gray}{$\pm$ 0.000}} &
    0.060{\scriptsize \color{gray}{$\pm$ 0.000}} &
    0.029{\scriptsize \color{gray}{$\pm$ 0.002}} &
    0.027{\scriptsize \color{gray}{$\pm$ 0.002}} & \cc
    \underline{0.024}{\scriptsize \color{gray}{$\pm$ 0.001}} & \cc
    \underline{0.024}{\scriptsize \color{gray}{$\pm$ 0.001}} & \cc
    0.027{\scriptsize \color{gray}{$\pm$ 0.002}} & \cc
    \textbf{0.023}{\scriptsize \color{gray}{$\pm$ 0.001}} \\
    $1$ & DSC~$\uparrow$ &
    0.833{\scriptsize \color{gray}{$\pm$ 0.000}} &
    0.756{\scriptsize \color{gray}{$\pm$ 0.000}} &
    0.872{\scriptsize \color{gray}{$\pm$ 0.012}} &
    0.867{\scriptsize \color{gray}{$\pm$ 0.014}} & \cc
    0.874{\scriptsize \color{gray}{$\pm$ 0.012}} & \cc
    0.873{\scriptsize \color{gray}{$\pm$ 0.011}} & \cc
    \textbf{0.885}{\scriptsize \color{gray}{$\pm$ 0.011}} & \cc
    \underline{0.883}{\scriptsize \color{gray}{$\pm$ 0.012}} \\
    & HD~$\downarrow$ &
    51.64{\scriptsize \color{gray}{$\pm$ 0.00}} &
    54.30{\scriptsize \color{gray}{$\pm$ 0.00}} &
    44.59{\scriptsize \color{gray}{$\pm$ 4.66}} &
    \underline{44.41}{\scriptsize \color{gray}{$\pm$ 4.74}} & \cc
    \textbf{42.68}{\scriptsize \color{gray}{$\pm$ 4.82}} & \cc
    47.10{\scriptsize \color{gray}{$\pm$ 4.89}} & \cc
    48.14{\scriptsize \color{gray}{$\pm$ 4.87}} & \cc
    45.14{\scriptsize \color{gray}{$\pm$ 4.89}} \\

    \cmidrule[0.2pt]{2-10}

    \textit{minor} &
    PSNR~$\uparrow$ &
    21.36{\scriptsize \color{gray}{$\pm$ 0.00}} &
    21.08{\scriptsize \color{gray}{$\pm$ 0.00}} &
    22.56{\scriptsize \color{gray}{$\pm$ 0.55}} &
    22.99{\scriptsize \color{gray}{$\pm$ 0.55}} & \cc
    23.23{\scriptsize \color{gray}{$\pm$ 0.34}} & \cc
    \underline{23.44}{\scriptsize \color{gray}{$\pm$ 0.33}} & \cc
    23.28{\scriptsize \color{gray}{$\pm$ 0.36}} & \cc
    \textbf{23.63}{\scriptsize \color{gray}{$\pm$ 0.43}} \\
    \textit{atrophy} &
    SSIM~$\uparrow$ &
    0.599{\scriptsize \color{gray}{$\pm$ 0.000}} &
    0.586{\scriptsize \color{gray}{$\pm$ 0.000}} &
    0.662{\scriptsize \color{gray}{$\pm$ 0.023}} &
    0.657{\scriptsize \color{gray}{$\pm$ 0.024}} & \cc
    0.682{\scriptsize \color{gray}{$\pm$ 0.018}} & \cc
    0.685{\scriptsize \color{gray}{$\pm$ 0.018}} & \cc
    \textbf{0.693}{\scriptsize \color{gray}{$\pm$ 0.018}} & \cc
    \underline{0.687}{\scriptsize \color{gray}{$\pm$ 0.019}}\\
    \textit{growth} &
    MAE~$\downarrow$ &
    0.141{\scriptsize \color{gray}{$\pm$ 0.000}} &
    0.147{\scriptsize \color{gray}{$\pm$ 0.000}} &
    0.121{\scriptsize \color{gray}{$\pm$ 0.007}} &
    0.114{\scriptsize \color{gray}{$\pm$ 0.007}} & \cc
    0.110{\scriptsize \color{gray}{$\pm$ 0.005}} & \cc
    \underline{0.108}{\scriptsize \color{gray}{$\pm$ 0.004}} & \cc
    0.109{\scriptsize \color{gray}{$\pm$ 0.005}} & \cc
    \textbf{0.106}{\scriptsize \color{gray}{$\pm$ 0.005}}\\
    $2$ & MSE~$\downarrow$ &
    0.038{\scriptsize \color{gray}{$\pm$ 0.000}} &
    0.042{\scriptsize \color{gray}{$\pm$ 0.000}} &
    0.027{\scriptsize \color{gray}{$\pm$ 0.003}} &
    0.024{\scriptsize \color{gray}{$\pm$ 0.002}} &\cc
    0.021{\scriptsize \color{gray}{$\pm$ 0.002}} & \cc
    \textbf{0.020}{\scriptsize \color{gray}{$\pm$ 0.002}} & \cc
    0.021{\scriptsize \color{gray}{$\pm$ 0.002}} & \cc
    \textbf{0.020}{\scriptsize \color{gray}{$\pm$ 0.002}}\\
    & DSC~$\uparrow$ &
    0.900{\scriptsize \color{gray}{$\pm$ 0.000}} &
    0.874{\scriptsize \color{gray}{$\pm$ 0.000}} &
    \textbf{0.949}{\scriptsize \color{gray}{$\pm$ 0.004}} &
    \textbf{0.949}{\scriptsize \color{gray}{$\pm$ 0.004}} & \cc
    0.936{\scriptsize \color{gray}{$\pm$ 0.009}} & \cc
    0.939{\scriptsize \color{gray}{$\pm$ 0.007}} & \cc
    0.948{\scriptsize \color{gray}{$\pm$ 0.005}} & \cc
    0.948{\scriptsize \color{gray}{$\pm$ 0.006}} \\
    & HD~$\downarrow$ &
    38.15{\scriptsize \color{gray}{$\pm$ 0.00}} &
    41.67{\scriptsize \color{gray}{$\pm$ 0.00}} &
    35.74{\scriptsize \color{gray}{$\pm$ 5.67}} &
    \textbf{29.40}{\scriptsize \color{gray}{$\pm$ 4.77}} & \cc
    34.59{\scriptsize \color{gray}{$\pm$ 6.20}} & \cc
    39.86{\scriptsize \color{gray}{$\pm$ 6.40}} & \cc
    \underline{31.66}{\scriptsize \color{gray}{$\pm$ 5.21}} & \cc
    36.98{\scriptsize \color{gray}{$\pm$ 6.04}} \\
  
    \cmidrule[0.2pt]{2-10}

    \textit{major} &
    PSNR~$\uparrow$ &
    19.02{\scriptsize \color{gray}{$\pm$ 0.00}} &
    18.41{\scriptsize \color{gray}{$\pm$ 0.00}} &
    21.40{\scriptsize \color{gray}{$\pm$ 0.33}} &
    21.68{\scriptsize \color{gray}{$\pm$ 0.32}} & \cc
    21.94{\scriptsize \color{gray}{$\pm$ 0.34}} & \cc
    \underline{22.01}{\scriptsize \color{gray}{$\pm$ 0.33}} & \cc
    \underline{22.01}{\scriptsize \color{gray}{$\pm$ 0.30}} & \cc
    \textbf{22.10}{\scriptsize \color{gray}{$\pm$ 0.31}}\\
    \textit{atrophy} &
    SSIM~$\uparrow$ &
    0.468{\scriptsize \color{gray}{$\pm$ 0.000}} &
    0.420{\scriptsize \color{gray}{$\pm$ 0.000}} &
    \underline{0.607}{\scriptsize \color{gray}{$\pm$ 0.017}} &
    0.588{\scriptsize \color{gray}{$\pm$ 0.017}} & \cc
    \underline{0.607}{\scriptsize \color{gray}{$\pm$ 0.014}} & \cc
    0.606{\scriptsize \color{gray}{$\pm$ 0.014}} & \cc
    \underline{0.607}{\scriptsize \color{gray}{$\pm$ 0.014}} & \cc
    \textbf{0.613}{\scriptsize \color{gray}{$\pm$ 0.013}}\\
    \textit{growth} &
    MAE~$\downarrow$ &
    0.186{\scriptsize \color{gray}{$\pm$ 0.000}} &
    0.210{\scriptsize \color{gray}{$\pm$ 0.000}} &
    0.135{\scriptsize \color{gray}{$\pm$ 0.006}} &
    0.131{\scriptsize \color{gray}{$\pm$ 0.006}} & \cc
    0.129{\scriptsize \color{gray}{$\pm$ 0.006}} & \cc
    0.129{\scriptsize \color{gray}{$\pm$ 0.006}} & \cc
    \underline{0.128}{\scriptsize \color{gray}{$\pm$ 0.005}} & \cc
    \textbf{0.126}{\scriptsize \color{gray}{$\pm$ 0.005}}\\
    $3$ & MSE~$\downarrow$ &
    0.063{\scriptsize \color{gray}{$\pm$ 0.000}} &
    0.080{\scriptsize \color{gray}{$\pm$ 0.000}} &
    0.032{\scriptsize \color{gray}{$\pm$ 0.003}} &
    0.030{\scriptsize \color{gray}{$\pm$ 0.002}} & \cc
    0.028{\scriptsize \color{gray}{$\pm$ 0.002}} & \cc
    0.028{\scriptsize \color{gray}{$\pm$ 0.002}} & \cc
    \textbf{0.027}{\scriptsize \color{gray}{$\pm$ 0.002}} &\cc
    \textbf{0.027}{\scriptsize \color{gray}{$\pm$ 0.002}}\\
    & DSC~$\uparrow$ &
    0.762{\scriptsize \color{gray}{$\pm$ 0.000}} &
    0.631{\scriptsize \color{gray}{$\pm$ 0.000}} &
    0.784{\scriptsize \color{gray}{$\pm$ 0.016}} &
    0.779{\scriptsize \color{gray}{$\pm$ 0.019}} & \cc
    0.807{\scriptsize \color{gray}{$\pm$ 0.014}} & \cc
    0.803{\scriptsize \color{gray}{$\pm$ 0.012}} & \cc
    \textbf{0.817}{\scriptsize \color{gray}{$\pm$ 0.016}} &\cc
    \underline{0.814}{\scriptsize \color{gray}{$\pm$ 0.017}}\\
    & HD~$\downarrow$ &
    65.97{\scriptsize \color{gray}{$\pm$ 0.00}} &
    67.73{\scriptsize \color{gray}{$\pm$ 0.00}} &
    61.43{\scriptsize \color{gray}{$\pm$ 7.26}} &
    60.36{\scriptsize \color{gray}{$\pm$ 7.37}} & \cc
    \textbf{51.28}{\scriptsize \color{gray}{$\pm$ 7.13}} & \cc
    54.79{\scriptsize \color{gray}{$\pm$ 7.19}} & \cc
    65.65{\scriptsize \color{gray}{$\pm$ 7.17}} & \cc
    \underline{53.81}{\scriptsize \color{gray}{$\pm$ 7.49}}\\
 
    \cmidrule[1pt]{1-10}

    \textbf{Brain} &
    PSNR~$\uparrow$ &
    30.07{\scriptsize \color{gray}{$\pm$ 0.00}} &
    29.56{\scriptsize \color{gray}{$\pm$ 0.00}} &
    31.55{\scriptsize \color{gray}{$\pm$ 0.20}} &
    31.57{\scriptsize \color{gray}{$\pm$ 0.23}} & \cc
    32.01{\scriptsize \color{gray}{$\pm$ 0.19}} & \cc
    32.34{\scriptsize \color{gray}{$\pm$ 0.20}} & \cc
    \underline{32.40}{\scriptsize \color{gray}{$\pm$ 0.20}} & \cc
    \textbf{32.41}{\scriptsize \color{gray}{$\pm$ 0.20}} \\
    \textbf{MS} &
    SSIM~$\uparrow$ &
    0.895{\scriptsize \color{gray}{$\pm$ 0.000}} &
    0.888{\scriptsize \color{gray}{$\pm$ 0.000}} &
    0.909{\scriptsize \color{gray}{$\pm$ 0.003}} &
    0.907{\scriptsize \color{gray}{$\pm$ 0.003}} & \cc
    0.914{\scriptsize \color{gray}{$\pm$ 0.002}} & \cc
    \textbf{0.915}{\scriptsize \color{gray}{$\pm$ 0.002}} & \cc
    0.913{\scriptsize \color{gray}{$\pm$ 0.002}} & \cc
    \textbf{0.915}{\scriptsize \color{gray}{$\pm$ 0.002}}\\
    \textbf{Images} &
    MAE~$\downarrow$ &
    0.028{\scriptsize \color{gray}{$\pm$ 0.000}} &
    0.030{\scriptsize \color{gray}{$\pm$ 0.000}} &
    0.024{\scriptsize \color{gray}{$\pm$ 0.000}} &
    0.024{\scriptsize \color{gray}{$\pm$ 0.001}} & \cc
    0.023{\scriptsize \color{gray}{$\pm$ 0.000}} & \cc
    \textbf{0.021}{\scriptsize \color{gray}{$\pm$ 0.000}} & \cc
    \textbf{0.021}{\scriptsize \color{gray}{$\pm$ 0.000}} & \cc
    \textbf{0.021}{\scriptsize \color{gray}{$\pm$ 0.000}} \\
    $4$ & MSE~$\downarrow$ &
    0.004{\scriptsize \color{gray}{$\pm$ 0.000}} &
    0.005{\scriptsize \color{gray}{$\pm$ 0.000}} &
    0.004{\scriptsize \color{gray}{$\pm$ 0.000}} &
    0.004{\scriptsize \color{gray}{$\pm$ 0.000}} & \cc
    \textbf{0.003}{\scriptsize \color{gray}{$\pm$ 0.000}} & \cc
    \textbf{0.003}{\scriptsize \color{gray}{$\pm$ 0.000}} & \cc
    \textbf{0.003}{\scriptsize \color{gray}{$\pm$ 0.000}} & \cc
    \textbf{0.003}{\scriptsize \color{gray}{$\pm$ 0.000}} \\
    & DSC~$\uparrow$ &
    0.739{\scriptsize \color{gray}{$\pm$ 0.000}} &
    0.682{\scriptsize \color{gray}{$\pm$ 0.000}} &
    0.774{\scriptsize \color{gray}{$\pm$ 0.007}} &
    0.771{\scriptsize \color{gray}{$\pm$ 0.007}} & \cc
    0.775{\scriptsize \color{gray}{$\pm$ 0.007}} & \cc
    \textbf{0.777}{\scriptsize \color{gray}{$\pm$ 0.007}} & \cc
    \textbf{0.777}{\scriptsize \color{gray}{$\pm$ 0.007}} & \cc
    0.774{\scriptsize \color{gray}{$\pm$ 0.007}} \\
    & HD~$\downarrow$ &
    22.73{\scriptsize \color{gray}{$\pm$ 0.00}} &
    26.23{\scriptsize \color{gray}{$\pm$ 0.00}} &
    22.00{\scriptsize \color{gray}{$\pm$ 1.30}} &
    \textbf{20.91}{\scriptsize \color{gray}{$\pm$ 1.23}} & \cc
    22.38{\scriptsize \color{gray}{$\pm$ 1.28}} & \cc
    21.72{\scriptsize \color{gray}{$\pm$ 1.16}} & \cc
    22.21{\scriptsize \color{gray}{$\pm$ 1.27}} & \cc
    \underline{21.28}{\scriptsize \color{gray}{$\pm$ 1.27}}\\

    \cmidrule[1pt]{1-10}

    \textbf{Brain} &
    PSNR~$\uparrow$ &
    35.32{\scriptsize \color{gray}{$\pm$ 0.00}} &
    33.60{\scriptsize \color{gray}{$\pm$ 0.00}} &
    35.73{\scriptsize \color{gray}{$\pm$ 0.13}} &
    35.49{\scriptsize \color{gray}{$\pm$ 0.17}} & \cc
    \underline{35.86}{\scriptsize \color{gray}{$\pm$ 0.12}} & \cc
    \textbf{35.90}{\scriptsize \color{gray}{$\pm$ 0.14}} & \cc
    35.77{\scriptsize \color{gray}{$\pm$ 0.12}} & \cc
    35.79{\scriptsize \color{gray}{$\pm$ 0.15}}\\
    \textbf{GBM} &
    SSIM~$\uparrow$ &
    0.929{\scriptsize \color{gray}{$\pm$ 0.000}} &
    0.895{\scriptsize \color{gray}{$\pm$ 0.000}} &
    0.935{\scriptsize \color{gray}{$\pm$ 0.001}} &
    \underline{0.940}{\scriptsize \color{gray}{$\pm$ 0.001}} & \cc
    \underline{0.940}{\scriptsize \color{gray}{$\pm$ 0.001}} & \cc
    \textbf{0.943}{\scriptsize \color{gray}{$\pm$ 0.001}} & \cc
    0.937{\scriptsize \color{gray}{$\pm$ 0.001}} & \cc
    0.939{\scriptsize \color{gray}{$\pm$ 0.001}} \\
    \textbf{Images} &
    MAE~$\downarrow$ &
    0.017{\scriptsize \color{gray}{$\pm$ 0.000}} &
    0.024{\scriptsize \color{gray}{$\pm$ 0.000}} &
    0.015{\scriptsize \color{gray}{$\pm$ 0.000}} &
    \textbf{0.014}{\scriptsize \color{gray}{$\pm$ 0.000}} & \cc
    \textbf{0.014}{\scriptsize \color{gray}{$\pm$ 0.000}} & \cc
    \textbf{0.014}{\scriptsize \color{gray}{$\pm$ 0.000}} & \cc
    0.015{\scriptsize \color{gray}{$\pm$ 0.000}} & \cc
    0.015{\scriptsize \color{gray}{$\pm$ 0.000}} \\
    $5$ & MSE~$\downarrow$ &
    0.002{\scriptsize \color{gray}{$\pm$ 0.000}} &
    0.005{\scriptsize \color{gray}{$\pm$ 0.000}} &
    \textbf{0.001}{\scriptsize \color{gray}{$\pm$ 0.000}} &
    0.002{\scriptsize \color{gray}{$\pm$ 0.000}} & \cc
    \textbf{0.001}{\scriptsize \color{gray}{$\pm$ 0.000}} & \cc
    \textbf{0.001}{\scriptsize \color{gray}{$\pm$ 0.000}} & \cc
    \textbf{0.001}{\scriptsize \color{gray}{$\pm$ 0.000}} & \cc
    \textbf{0.001}{\scriptsize \color{gray}{$\pm$ 0.000}} \\
    & DSC~$\uparrow$ &
    \underline{0.300}{\scriptsize \color{gray}{$\pm$ 0.000}} &
    0.287{\scriptsize \color{gray}{$\pm$ 0.000}} &
    0.258{\scriptsize \color{gray}{$\pm$ 0.018}} &
    0.253{\scriptsize \color{gray}{$\pm$ 0.017}} & \cc
    \textbf{0.302}{\scriptsize \color{gray}{$\pm$ 0.019}} & \cc
    0.266{\scriptsize \color{gray}{$\pm$ 0.018}} & \cc
    0.286{\scriptsize \color{gray}{$\pm$ 0.019}} & \cc
    0.287{\scriptsize \color{gray}{$\pm$ 0.017}} \\
    & HD~$\downarrow$ &
    \underline{170.44}{\scriptsize \color{gray}{$\pm$ 0.00}} &
    \textbf{165.62}{\scriptsize \color{gray}{$\pm$ 0.00}} &
    195.52{\scriptsize \color{gray}{$\pm$ 7.69}} &
    189.61{\scriptsize \color{gray}{$\pm$ 7.64}} & \cc
    198.19{\scriptsize \color{gray}{$\pm$ 7.78}} & \cc
    185.14{\scriptsize \color{gray}{$\pm$ 7.69}} & \cc
    196.37{\scriptsize \color{gray}{$\pm$ 7.74}} & \cc
    181.66{\scriptsize \color{gray}{$\pm$ 7.66}} \\

    \cmidrule[1pt]{1-10}
    $1,4,5$ & \textbf{Rank}~$\downarrow$ &
    {\large 6.3{\scriptsize \color{gray}{$\pm$ 1.6}}} &
    {\large 7.3{\scriptsize \color{gray}{$\pm$ 2.0}}} &
    {\large 4.9{\scriptsize \color{gray}{$\pm$ 1.4}}} &
    {\large 4.6{\scriptsize \color{gray}{$\pm$ 1.9}}} & \cc
    {\large 2.9{\scriptsize \color{gray}{$\pm$ 1.9}}} & \cc
    {\large \underline{2.3}{\scriptsize \color{gray}{$\pm$ 1.6}}} & \cc
    {\large 3.4{\scriptsize \color{gray}{$\pm$ 2.0}}} & \cc
    {\large \textbf{2.1}{\scriptsize \color{gray}{$\pm$ 1.3}}}\\

    $1,2,3,4,5$ & \textbf{Rank}~$\downarrow$ &
    {\large 6.5{\scriptsize \color{gray}{$\pm$ 1.3}}} &
    {\large 7.6{\scriptsize \color{gray}{$\pm$ 1.5}}} &
    {\large 4.9{\scriptsize \color{gray}{$\pm$ 1.5}}} &
    {\large 4.5{\scriptsize \color{gray}{$\pm$ 1.8}}} & \cc
    {\large 3.1{\scriptsize \color{gray}{$\pm$ 1.6}}} & \cc
    {\large \underline{2.7}{\scriptsize \color{gray}{$\pm$ 1.7}}} & \cc
    {\large 3.0{\scriptsize \color{gray}{$\pm$ 1.8}}} & \cc
    {\large \textbf{2.0}{\scriptsize \color{gray}{$\pm$ 1.2}}}\\

    \bottomrule

    \end{tabular*}
    }
\label{tab:main}
\end{table*}

\begin{figure*}[!tb]
    \centering
    \includegraphics[width=\textwidth]{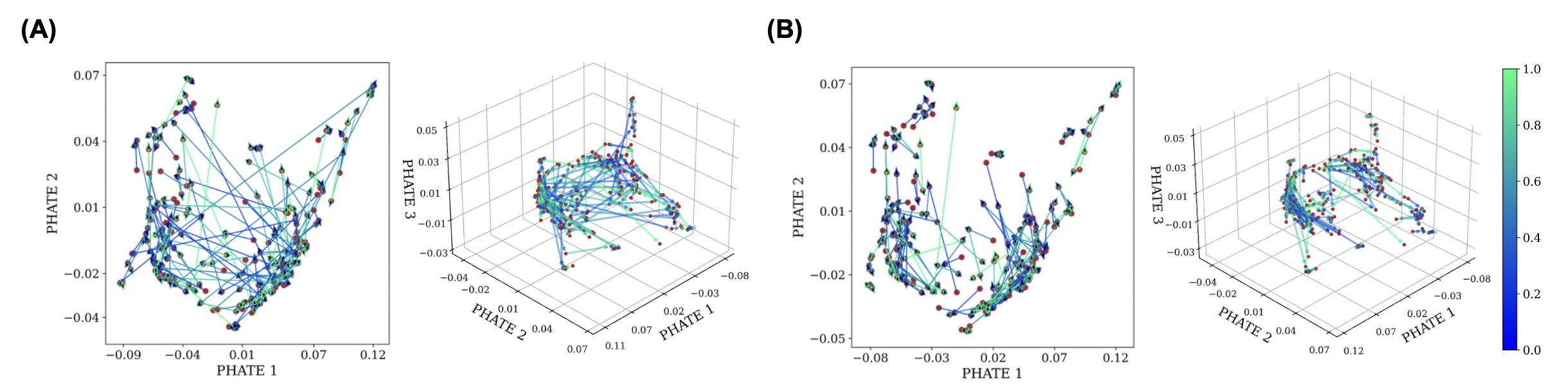}
    \caption{Joint representation space and the effect of contrastive learning regularization. Red dots are the observed disease states and arrows connect adjacent transitions. Normalized time is color coded. (A) Without regularization ($\lambda_c = 0$). (B) With contrastive learning regularization ($\lambda_c = 0.01$).}
    \label{fig:space_time}
\end{figure*}

Another notable phenomenon arises when we break down the retinal images into subsets. Compared to the other alternatives, our proposed methods show similar atrophy prediction performance (DSC, HD) for eyes with ``minor atrophy growth'', but significantly better for eyes with ``major atrophy growth''. Major/minor growth is defined by whether the ground truth masks differ by more than 0.1 in DSC. This implies that while the other methods may be on par with us in performing image reconstruction, our method is better at modeling the actual disease progression dynamics.

In Figure~\ref{fig:space_time}, we visualize the joint patient representations with and without contrastive learning regularization. The latent space of the ImageFlowNet bottleneck layer is visualized after projection into the 2D/3D PHATE space~\cite{PHATE}. Indeed, the contrastive loss helped organize better structures in the latent space, as is evident in fewer global-range connections and smoother transitions over time.

\subsection{Test-Time Optimization Improves Prediction Leveraging the Entire Patient History}
\label{sec:test_time_optimization}

While it only requires a single observation to infer the trajectory of a new patient, we can further improve ImageFlowNet performance if we use the entire patient history to perform test-time optimization. More specifically, we could take the trained ImageFlowNet and fine-tune the flow field $f_\theta$ with the previous measurements $\{x_1, x_2, ..., x_{n - 1}\}$, and then predict $x_n$ with the fine-tuned model.

\begin{wraptable}{r}{0.5\textwidth}
    \vskip -18pt
    \caption{Effect of test-time optimization.}
    \centering
    \scalebox{0.54}{
    \begin{tabular}{lllcccccc}
    \toprule
    Iterations & Learning Rate & PSNR$\uparrow$ & SSIM$\uparrow$ & MAE$\downarrow$ & MSE$\downarrow$ & DSC$\uparrow$ & HD$\downarrow$ \\
    \toprule
    N/A & N/A & 22.31 & 0.643 & 0.123 & 0.027 & 0.827 & 51.07 \\
    \specialrule{0.1pt}{0pt}{2pt}
    1 & $10^{-4}$ & 22.52 & \textbf{0.646} & 0.120 & \textbf{0.025} & \textbf{0.829} & \textbf{48.97} \\
    1 & $10^{-5}$ & 22.36 & 0.643 & 0.122 & 0.027 & 0.827 & 51.02 \\
    1 & $10^{-6}$ & 22.31 & 0.643 & 0.123 & 0.027 & 0.827 & 51.07 \\
    10 & $10^{-4}$ & 20.63 & 0.605 & 0.157 & 0.042 & 0.749 & 64.79 \\
    10 & $10^{-5}$ & \underline{22.59} & \textbf{0.646} & \textbf{0.119} & \textbf{0.025} & \textbf{0.829} & 49.92 \\
    10 & $10^{-6}$ & 22.36 & 0.644 & 0.122 & 0.027 & 0.827 & 51.01 \\
    100 & $10^{-4}$ & 19.63 & 0.571 & 0.177 & 0.056 & 0.726 & 70.12 \\
    100 & $10^{-5}$ & 20.92 & 0.614 & 0.152 & 0.040 & 0.759 & 58.76 \\
    100 & $10^{-6}$ & \textbf{22.61} & \textbf{0.646} & \textbf{0.119} & \textbf{0.025} & \textbf{0.829} & \underline{49.74} \\
    \bottomrule
    \end{tabular}
    }
\label{tab:test_time_optimization}
\end{wraptable}

We investigated the effect of test-time optimization using longitudinal series with at least 3 images from the retinal image dataset. The results for ImageFlowNet\textsubscript{ODE} are summarized in Table~\ref{tab:test_time_optimization}, and similar trends are observed in other model variants. This indicates the possibility of trading computations for performance when the patient's history is accessible.

\clearpage
\subsection{
\texorpdfstring{Modeling Alternative Trajectories from the Same Starting Point with ImageFlowNet\textsubscript{SDE}}{Modeling Alternative Trajectories from the Same Starting Point with ImageFlowNetSDE}
}

\begin{wrapfigure}{l}{0.7\textwidth}
\vskip -12pt
\centering
\includegraphics[width=0.64\textwidth]{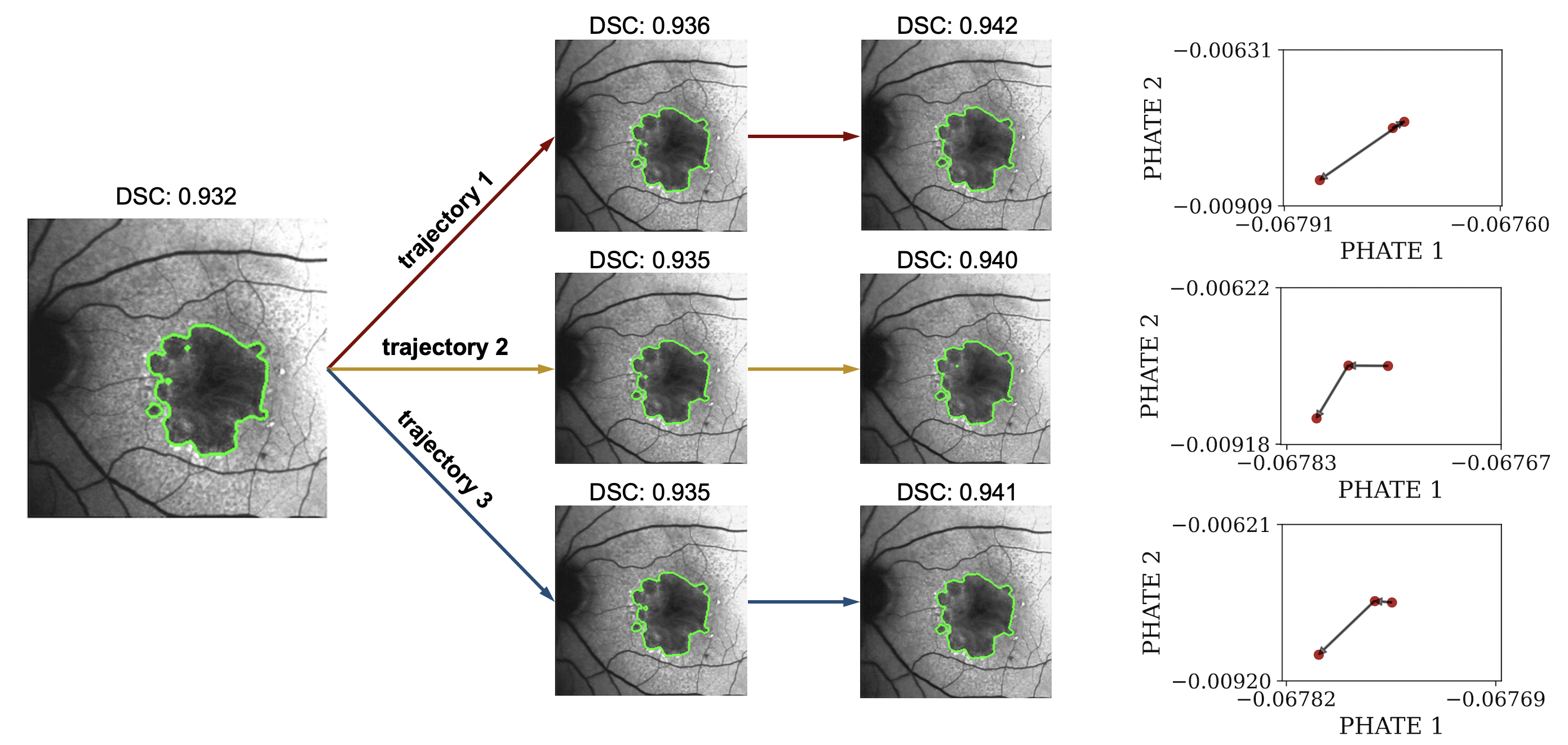}
\vskip -6pt
\caption{ImageFlowNet\textsubscript{SDE} alternative trajectories. Multiple inferences predict non-identical disease progressions, and their vectors in the joint representations space indeed follow different trajectories.}
\label{fig:multiple_trajectories}
\end{wrapfigure}

In Figure~\ref{fig:multiple_trajectories}, we demonstrate ImageFlowNet\textsubscript{SDE}'s ability to model alternative trajectories. The model infers several trajectories from the same initial image, as indicated by the varied predicted disease regions and the different representation vectors in the same PHATE space. The minimal variation among these inferences could stem from the absence of explicit encouragement to produce highly divergent trajectories during training, which might be an interesting direction for future research.

\subsection{Ablations Studies}

\begin{wraptable}{r}{0.5\textwidth}
\vskip -24pt
\centering
\caption{Flow field formulation.}
\label{tab:gradient_field}
\scalebox{0.7}{
    \begin{tabular}{lcccccc}
    \toprule
    & PSNR$\uparrow$ & SSIM$\uparrow$ & MAE$\downarrow$ & MSE$\downarrow$ & DSC$\uparrow$ & HD$\downarrow$ \\
    \toprule
    $f_\theta(z_t, t)$ & 22.42 & 0.643 & 0.123 & 0.027 & 0.872 & 48.38 \\
    \rowcolor{Gainsboro!60}
    $f_\theta(z_t)$ & \textbf{22.63} & \textbf{0.646} & \textbf{0.119} & \textbf{0.024} & \textbf{0.874} & \textbf{42.68} \\
    \bottomrule
    \end{tabular}
}
\caption{Latent representation.}
\label{tab:multiscale_ode}
\scalebox{0.56}{
    \begin{tabular}{lcccccc}
    \toprule
    & PSNR$\uparrow$ & SSIM$\uparrow$ & MAE$\downarrow$ & MSE$\downarrow$ & DSC$\uparrow$ & HD$\downarrow$ \\
    \toprule
    bottleneck only & 22.33 & 0.639 & 0.122 & 0.026 & 0.850 & 48.13 \\
    all unique resolutions & 22.49 & 0.643 & 0.122 & 0.025 & 0.859 & 43.39 \\
    \rowcolor{Gainsboro!60}
    all unique layers & \textbf{22.63} & \textbf{0.646} & \textbf{0.119} & \textbf{0.024} & \textbf{0.874} & \textbf{42.68} \\
    \bottomrule
    \end{tabular}
    }
\caption{Visual feature regularization.}
\label{tab:regu_visual_feature}
\scalebox{0.7}{
    \begin{tabular}{lcccccc}
    \toprule
    $\lambda_v$ & PSNR$\uparrow$ & SSIM$\uparrow$ & MAE$\downarrow$ & MSE$\downarrow$ & DSC$\uparrow$ & HD$\downarrow$ \\
    \toprule
    $0$ & 22.63 & 0.646 & 0.119 & \textbf{0.024} & \textbf{0.874} & \textbf{42.68} \\
    \rowcolor{Gainsboro!60}
    $0.001$ & \textbf{22.65} & \textbf{0.658} & \textbf{0.118} & \textbf{0.024} & 0.872 & 44.27 \\
    $0.01$ & 22.64 & 0.650 & 0.120 & 0.025 & 0.872 & 45.89 \\
    $0.1$ & 22.57 & 0.647 & 0.120 & 0.025 & 0.869 & 50.69 \\
    $1$ & 22.54 & 0.634 & 0.124 & 0.027 & 0.867 & 48.13 \\
    \bottomrule
    \end{tabular}
}
\caption{Contrastive regularization.}
\label{tab:regu_contrastive}
\scalebox{0.7}{
    \begin{tabular}{lcccccc}
    \toprule
    $\lambda_c$ & PSNR$\uparrow$ & SSIM$\uparrow$ & MAE$\downarrow$ & MSE$\downarrow$ & DSC$\uparrow$ & HD$\downarrow$ \\
    \toprule
    $0$ & 22.63 & 0.646 & 0.119 & \textbf{0.024} & 0.874 & 42.68 \\
    $0.001$ & 22.63 & 0.646 & 0.119 & 0.025 & 0.872 & 46.23\\
    \rowcolor{Gainsboro!60}
    $0.01$ & \textbf{22.65} & \textbf{0.652} & \textbf{0.118} & \textbf{0.024} & \textbf{0.875} & \textbf{42.18}\\
    $0.1$ & 22.38 & 0.651 & 0.121 & 0.025 & 0.871 & 45.30 \\
    $1$ & 22.25 & 0.644 & 0.121 & 0.025 & 0.868 & 46.85 \\
    \bottomrule
    \end{tabular}
}
\caption{Smoothness regularization.}
\label{tab:regu_smoothness}
\scalebox{0.7}{
    \begin{tabular}{lcccccc}
    \toprule
    $\lambda_s$ & PSNR$\uparrow$ & SSIM$\uparrow$ & MAE$\downarrow$ & MSE$\downarrow$ & DSC$\uparrow$ & HD$\downarrow$ \\
    \toprule
    $0$ & 22.63 & 0.646 & 0.119 & \textbf{0.024} & 0.874 & \textbf{42.68} \\
    $0.001$ & 22.38 & 0.649 & 0.123 & 0.027 & 0.870 & 46.91 \\
    $0.01$ & 22.65 & 0.648 & 0.119 & 0.024 & 0.870 & 45.71 \\
    \rowcolor{Gainsboro!60}
    $0.1$ & \textbf{22.70} & \textbf{0.657} & \textbf{0.118} & \textbf{0.024} & \textbf{0.878} & 47.44 \\
    $1$ & 22.69 & 0.655 & \textbf{0.118} & \textbf{0.024} & 0.875 & 45.16 \\
    \bottomrule
    \end{tabular}
    }
\vskip -30pt
\end{wraptable}

We performed ablations on (1) time vs. position parameterization of ODE, (2) single-scale vs. multiscale ODE, and (3) effects of 3 regularizations.

\paragraph{Flow Field Formulation}
\label{sec:gradient_field_formulation}
Previously, we decided on the $f_\theta(z_t)$ formulation analytically, and here we support our decision with empirical evidence~(Table~\ref{tab:gradient_field}).

\paragraph{Single-scale vs Multiscale ODEs}
The UNet architecture uses hierarchical hidden layers to extract multiscale representations. Starting at the image resolution and ending at the bottleneck layer (bottom of the ``U''), the model produces increasingly higher-level and more global representations. In this study, we analyze the advantages of multiscale ODEs. Moreover, there might be multiple hidden layers at the same resolution. On which representations should we perform trajectory inference?

To study this, we explored the following settings: (1) infer a single-scale $z_t$ from the bottleneck layer, (2) infer multiscale $\{ z_t \}$ at all layers and use distinct $f_\theta$ for each resolution, but all hidden layers of the same resolution share the same $f_\theta$, and (3) infer multiscale $\{ z_t \}$ at all layers and use distinct $f_\theta$ for each hidden layer. The empirical results as shown in Table~\ref{tab:multiscale_ode} indicate that modeling all representations separately would lead to the best performance.

Note: To avoid confusion, all of these hidden layers produce outputs that are bridged by skip connections from the contraction path to the expansion path.

\paragraph{Effects of Regularizations}
We ablated visual feature regularization (Table~\ref{tab:regu_visual_feature}), contrastive learning regularization (Table~\ref{tab:regu_contrastive}), and trajectory smoothness regularization (Table~\ref{tab:regu_smoothness}).

\section{Related Works}

\paragraph{Disease Progression Modeling in Longitudinal Medical Images}
Most existing methods for modeling disease progression operate in the vector space of hand-selected features. The event-based model~(EBM)~\cite{EBM1, EBM2} and discriminative EBM~\cite{DEBM} use bivariate mixture modeling and univariate Gaussian modeling, respectively. Both methods analyze disease progression at the group level and do not predict the outcome for individuals, the need of which led to sequence modeling techniques. Liu et al.~\cite{Disease_prog_xgboost} used XGBoost~\cite{XGBoost} to model the progression of breast cancer. Lei et al.~\cite{Disease_prog_polynomial_network} used a polynomial network to describe selected image statistics of patients with Alzheimer's disease. LSTM~\cite{LSTM} and Transformer~\cite{Transformer} have also been recruited to predict the progression of diseases such as Alzheimer's~\cite{Disease_prog_lstm_AD, Disease_prog_transformer1} and COVID-19~\cite{Disease_prog_lstm_covid19} in longitudinal medical images.

\paragraph{Disease Progression Modeling in the Image Space}
Due to the challenges discussed in the Introduction, very few works tackle the disease progression problem in the image space. STLSTM~\cite{STLSTM} and LDDMM~\cite{LDDMM} use LSTMs to model time evolution, which is a sequential model that does not handle irregular sampling over time --- a limitation commonly seen in video prediction models as well~\cite{VPN, PredRNN, MCNet, Lumiere}. ManifoldExtrap~\cite{ImageDiseaseProg_manifold} projects images to a StyleGAN latent space and performs a linear walk whose direction and distance are determined by the nearest neighbor found in this latent space. This method only models the transition between the baseline visit and the follow-up visit and does not model the continuous-time evolution. Upgrading the latent space navigation (e.g. using ODEs to model continuous time) might be a good solution, which we leave to future investigations. Lachinov and his collaborators present an approach~\cite{lachinov2023learning, mai2024deep} most similar to ours, as they also aim to model disease progression at the image level. However, their method is limited to segmentation output and does not explore multiscale trajectories, our innovative differential equation formulations, or our comprehensive suite of regularizations. These key elements distinguish our work from theirs.

\paragraph{Neural ODEs for Disease Progression Modeling}
To better accommodate irregular sampling, a common situation in longitudinal medical images, researchers later investigated differential equation models~\cite{Disease_prog_diffeq}, especially neural ODEs~\cite{NeuralODE}. Neural ODEs have been successfully applied to predict the dynamics of individual patients with Alzheimer's~\cite{ODE_alzheimer} and COVID-19~\cite{ODE_covid19}, but exclusively on biomakers and/or attributes extracted from images rather than on the medical images themselves.

\paragraph{Applying Neural ODEs to Images}
At the birth of neural ODEs~\cite{NeuralODE}, the possibility of applying them to images was discussed, but only as a drop-in replacement for residual blocks in a ResNet model, such as for image classification~\cite{ODE_image_cls}. UNode~\cite{UNode} and follow-ups~\cite{UNode_followup1, UNode_followup2} adapted it to work on an image-to-image task for the first time on image segmentation, but they treated the ODE as additional trainable parameters beyond convolution blocks and \textit{did not exploit its ability to model time}. Among similar endeavors, we are the first to use neural ODEs in the natural and intuitive manner by actually modeling how latent representations evolve over time in spatial-temporal data.

\section{Discussion on Clinical Implications}

ImageFlowNet offers two key clinically-relevant benefits. (1) It implicitly models all image features. Since progression is modeled in the image space, users can retrospectively analyze the trajectory of any image feature post-training without explicitly defining and modeling them as individual variables prior to training. (2) It provides intuitive visualizations of disease progression. For example, in retinal geographic atrophy, doctors can visually demonstrate to patients how much their eyes would worsen if left untreated.

\section{Conclusion}
We introduced ImageFlowNet, a deep learning framework which uses joint representation spaces and multiscale position-parameterized differential equations to infer trajectories in irregularly-sampled longitudinal images. We provided theoretical evidence to support its soundness and demonstrated its empirical effectiveness across three longitudinal medical image datasets. We believe that our method offers a promising approach to image-level trajectory analysis that can model progression in medical image datasets, a relatively underexplored yet highly promising area of research.

\section*{Limitations and Broader Impacts}
\subsection*{Limitations}
We have not yet studied the capabilities of using ImageFlowNet in a clinical context for patient diagnosis, which we plan to cover in a follow-up study. Meanwhile, two main areas for future improvement are computational complexity and latent space organization. (1) Modeling multiscale dynamics through neural differential equations requires significant GPU space. Our experiments can be run on a single 24 GB GPU, but it will pose challenges when we extend to dynamics of images higher than 256 $\times$ 256 resolution. For those applications, reducing feature dimension, using fewer scales, or finding more efficient computations may be necessary. (2) There are still room for improving the organization of the joint latent space. Future researchers can use additional patient-level information or image-level features to further organize the latent space according to disease state and severity.

\subsection*{Broader Impacts} 
Our work can help us understand spatial-temporal systems, state transitions in longitudinal images, and in particular disease progression in longitudinal medical images. To the best of our knowledge, our work has no negative societal impact.

\section*{Acknowledgements}
This work was supported in part by the National Science Foundation (NSF Career Grant~2047856) and the National Institute of Health (NIH~1R01GM130847-01A1, NIH~1R01GM135929-01).

\bibliographystyle{unsrt}
\bibliography{references}

\renewcommand{\thetable}{S\arabic{table}}
\renewcommand{\thefigure}{S\arabic{figure}}
\renewcommand{\theHtable}{S\arabic{table}}
\renewcommand{\theHfigure}{S\arabic{figure}}
\setcounter{figure}{0}

\newpage
\renewcommand\appendixpagename{\centering\noindent\rule{\textwidth}{2pt} \LARGE Technical Appendices for\\ \inserttitle\\ \normalsize \noindent\rule{\textwidth}{1pt}}

\begin{appendices}



\Large \textbf{Table of Contents} \normalsize

\startcontents[sections]
\printcontents[sections]{l}{1}{\setcounter{tocdepth}{2}}

\clearpage
\newpage

\section{Propositions and Proofs}
\label{supp:proofs}

\subsection{Proposition~\ref{prop:expressiveness_equivalence}}

\begin{prop}
Let $f_{\theta}$ be a continuous function that satisfies the Lipschitz continuity and linear growth conditions. Also, let the initial state $y(t_0) = y_0$ satisfy the finite second moment requirement $(\mathbb{E}[|y(t_0)|^2] < \infty)$. Suppose $z(t_0)$ is the latent representation learned by ImageFlowNet at the initial state corresponding to $t_0$. Then, our neural ODEs $($Eqn~\eqref{eqn:imageflownet_diffeq}$)$ are at least as expressive as the original neural ODEs $($Eqn~\eqref{eqn:neural_ode_diffeq}$)$, and their solutions capture the same dynamics.
\end{prop}
    
We recall the two dynamic systems for original neural ODEs and our ODEs: 

Original neural ODEs:
\begin{equation*}
\frac{\mathrm{d} y(\tau)}{\mathrm{d}\tau} = f_\theta(y(\tau), \tau), \quad f_{\theta}: \mathbb{R}^n \times [0,T] \rightarrow \mathbb{R}^n
\end{equation*}

Our neural ODEs, with (1) superscript $\cdot^{(b)}$ omitted without loss of generality, (2) $z_\tau$ equivalently replaced by $z(\tau)$ for notation consistency, and (3) $f_\theta$ replaced by $\tilde{f}_\theta$ for distinction:
\begin{equation*}
\frac{\mathrm{d} z(\tau)}{\mathrm{d}\tau} = \tilde{f}_\theta(z(\tau)), \quad \tilde{f}_{\theta}: \mathbb{R}^m \rightarrow \mathbb{R}^m
\end{equation*}

\begin{proof}

\begin{theorem}[Picard-Lindel\"of~\cite{Book_DE}]
Let $D \subset \mathbb{R}^n$ be an open set, and let $f: D \times [0, T] \to \mathbb{R}^n$ be a continuous function that satisfies a Lipschitz condition in $y$ uniformly in $\tau$. Then, for any initial condition $y(t_0) = y_0$, there exists a unique solution to the initial value problem:
\[
\frac{\mathrm{d}y(\tau)}{\mathrm{d}\tau} = f(y(\tau), \tau), \quad y(t_0) = y_0.
\]
\end{theorem}

\textbf{Lipschitz Condition:}
$$
|f_\theta(y_1, \tau) - f_\theta(y_2, \tau)| \leq L |y_1 - y_2|
$$

\textbf{Linear Growth Condition:}
$$
|f_\theta(y, \tau)| \leq K (1 + |y|)
$$

Given these conditions, both the original neural ODE and the Latent Space Neural ODE have unique strong solutions.

Since both the original ODE and the Latent Space Neural ODE have unique solutions, we could then construct a bijective and sufficiently smooth mapping $ h: \mathbb{R}^n \times [0, T] \to \mathbb{R}^m $ such that $ z(\tau) = h(y(\tau), \tau) $.

We define a function $ h(y, \tau) $ that maps the state $ y(\tau) $ and time $ \tau $ to a new latent state $ z(\tau)$ as

$$h(y, \tau) := y(\tau) \oplus \tau ,$$

where $ \oplus $ denotes the concatenation of the state and time.

Then, as $h$ is bijective, the inverse function $ h^{-1} $ maps $ z(\tau) $ back to $ y(\tau) $ and $ \tau $. Given $ h(y, \tau) = y \oplus \tau $, the inverse is:

$$
h^{-1}(z) = \left(y(z_{\text{time}}), z_{\text{time}}\right)
$$

By the chain rule, the derivative of $ z(\tau) $ with respect to $\tau$ is:

$$
\frac{\mathrm{d}z(\tau)}{\mathrm{d}\tau} = \frac{\partial h}{\partial y} \frac{\mathrm{d}y}{\mathrm{d}\tau} + \frac{\partial h}{\partial \tau}
$$

Substituting the ODE for $ y(\tau) $, we get:

$$
\frac{\mathrm{d}z(\tau)}{\mathrm{d}\tau} = \frac{\partial h}{\partial y} f_\theta(y(\tau), \tau) + \frac{\partial h}{\partial \tau}
$$

We can then simply define the function $ \tilde{f}_\theta $ in the latent space such that it incorporates the dynamics from the original space:

$$
\tilde{f}_\theta(z(\tau)) := \frac{\partial h}{\partial y} f_\theta(y(\tau), \tau) + \frac{\partial h}{\partial \tau}
$$

The universal approximation theorem ensures that there exists a neural network parameterized by $ \theta $ that can approximate any continuous function, including $\tilde{f}_\theta(z(\tau))$.

\paragraph{Existence of Equivalent Function}

Since the neural network can approximate $\tilde{f}_\theta(z(\tau))$, there exists a function $\tilde{f}_\theta(z(\tau))$ in the latent space that can represent the same system behavior governed by $f_{\theta}(y(\tau),\tau)$ in the original space.

\paragraph{Proving Equivalence:}
Given $z(\tau) = h(y(\tau), \tau)$ and the corresponding functions $f_\theta$ and $\tilde{f}_\theta$, we have shown that the new ODE formulation:
$$
\frac{\mathrm{d}z(\tau)}{\mathrm{d}\tau} = \tilde{f}_\theta(z(\tau))
$$
captures the same dynamics as the original ODE:
$$
\frac{\mathrm{d}y(\tau)}{\mathrm{d}\tau} = f_\theta(y(\tau), \tau)
$$
\end{proof}

\clearpage
\newpage
\subsection{Proposition~\ref{prop: equivalence_DOT}}
\begin{prop}
If we consider an image as a distribution over a 2D grid, ImageFlowNet is equivalently solving a dynamic optimal transport problem, as it meets three essential criteria: (1) matching the density, (2) smoothing the dynamics, and (3) minimizing the transport cost, where the ground distance is the Euclidean distance in the latent joint embedding space. \end{prop}

\begin{proof}

ImageFlowNet can alternatively be viewed in the context of a dynamic optimal transport framework, which aims to determine the optimal plan $\pi$ to transport mass from an initial distribution $\mu$ to a target distribution $\nu$ for a fixed state interval $[\tau_i, \tau_j]$. The task meets three requirements of dynamic optimal transport: (1) matching the density, (2) smoothing the dynamics, and (3) minimizing the transport cost. The ground distance in the latent joint embedding space is the Euclidean distance.

\paragraph{Matching the density}

The image is a 2D grid, and the distribution for the pixel intensities is $\mu$ at $\tau_i$ and $\nu$ at $\tau_j$ on this grid. $\mu$ and $\nu$ are defined on measure space $\mathcal{X} \subset \mathbb{R}^n $ and $\mathcal{Y} \subset \mathbb{R}^n$ respectively. The set of all joint probability measures on $\mathcal{X} \times 
\mathcal{Y}$ is denoted as $\Pi(\mu,\nu)$ and $c(x,y)$ is the cost of moving a mass unit from the original distribution $\mu$ at state $\tau_i$ to the target distribution $\nu$ at state $\tau_j$. Then, the distance between the two distributions $\mu$ and $\nu$ is the p-Wasserstein distance:

\vskip -12pt
\begin{equation*}
W(\mu,\nu)_p := \left( \inf_{\pi \in \Pi(\mu,\nu)} \int_{\mathcal{X} \times \mathcal{Y}} c(x,y) d\pi(x,y)\right)^{\frac{1}{p}}, \text{ where } p \geq 1
\end{equation*}
\vskip -2pt

Benamou \& Brenier~\cite{BenamouAndBrenier} present a dynamic view of optimal transport, which links to differential equations. For the state interval $[\tau_i, \tau_j]$, there is a smooth and status-dependent density $P(z,\tau) \geq 0$ with $\int_{\mathbb{R}^n} P(z,\tau) dz = 1, \forall \tau \in [\tau_i, \tau_j]$, and a velocity fields $f(z, \tau)$ that obeys the continuity equation:

\vskip -8pt
\begin{equation*}
\partial_{\tau}P + \nabla \cdot (Pf) = 0, \text{ with } \tau \in [\tau_i, \tau_j] \text{ and } z \in \mathbb{R}^n \text{, where } P(\cdot, \tau_i) = \mu, \text{ } P(\cdot, \tau_j) = \nu
\end{equation*}
\vskip -4pt

\paragraph{Smoothing the dynamics}

The velocity fields $f(z,\tau)$ follows the Lipschitz condition $ |f(z_1, \tau) - f(z_2, \tau)| \leq L |z_1 - z_2| \text{ where } L > 0$, which ensures a smooth and controlled transport process. With the following setup, Benamou \& Brenier~\cite{BenamouAndBrenier} show that the Wasserstein distance with order 2 ($W_2$) is:

\vskip -10pt
\begin{equation*}
W(\mu,\nu)_2^{2} = \inf_{(p,f)} \int_{\mathbb{R}^n}\int_{\tau_i}^{\tau_j} P(z,\tau)\|f(z,\tau)\|^2 d\tau dz     
\end{equation*}
\vskip -4pt

\paragraph{Minimizing the transport cost}

Based on the main theorems in~\cite{TrajectoryNet, MIOFlow}, this problem aims to find the trajectory $f$ that minimizes the transport cost on the path space $\mathbb{R}^n$, we define the ground distance in the latent joint embedding space to be the Euclidean distance: 

\vskip -8pt
\begin{equation*}
W(\mu,\nu)_2^{2} = \inf_f \mathbb{E}\left[ \int_{\tau_i}^{\tau_j} \|f(z_{\tau},\tau) \|^2d\tau \right] \text{ s.t. } \frac{\partial z(\tau)}{\partial \tau} = f_{\theta}(z_{\tau}, \tau), \text{ } z_{\tau_i}  \sim \mu, \text{ } z_{\tau_j}  \sim \nu    
\end{equation*}
\vskip -4pt

Here, $f_{\theta}$ follows the ODE or the SDE. 

With the above setups, ImageFlowNet is equivalent to a dynamic optimal transport problem trying to match the density at different states. 

\end{proof}

\clearpage
\newpage
\section{Additional Background on Longitudinal Image Data}
Longitudinal image datasets, including but not limited to retinal images or even medical images, often come with several challenges: \fancynumber{1} high dimensionality, \fancynumber{2} temporal sparsity, \fancynumber{3} sampling irregularity, and \fancynumber{4} spatial misalignment.

\fancynumber{1} \textbf{High Dimensionality} is intrinsic to image data. For images with height of $H$ pixels, width of $W$ pixels and $C$ image channels, the dimensionality of the data is $\mathbb{R}^{H \times W \times C}$, which can easily go beyond a hundred thousand dimensions: a small image of $256 \times 256 \times 3$ has 196.6 thousand dimensions. Such high dimensionality is rarely encountered by most methods in time series prediction and temporal dynamics modeling~\cite{MIOFlow, tong2023simulation, nguyen2023causal, zapatero2023trellis, DynGFN, chen2024similarity, xiao2024xtsformer}.

\label{supp:additional_background}
\begin{figure*}[!thb]
    \centering
    \includegraphics[width=0.8\textwidth]{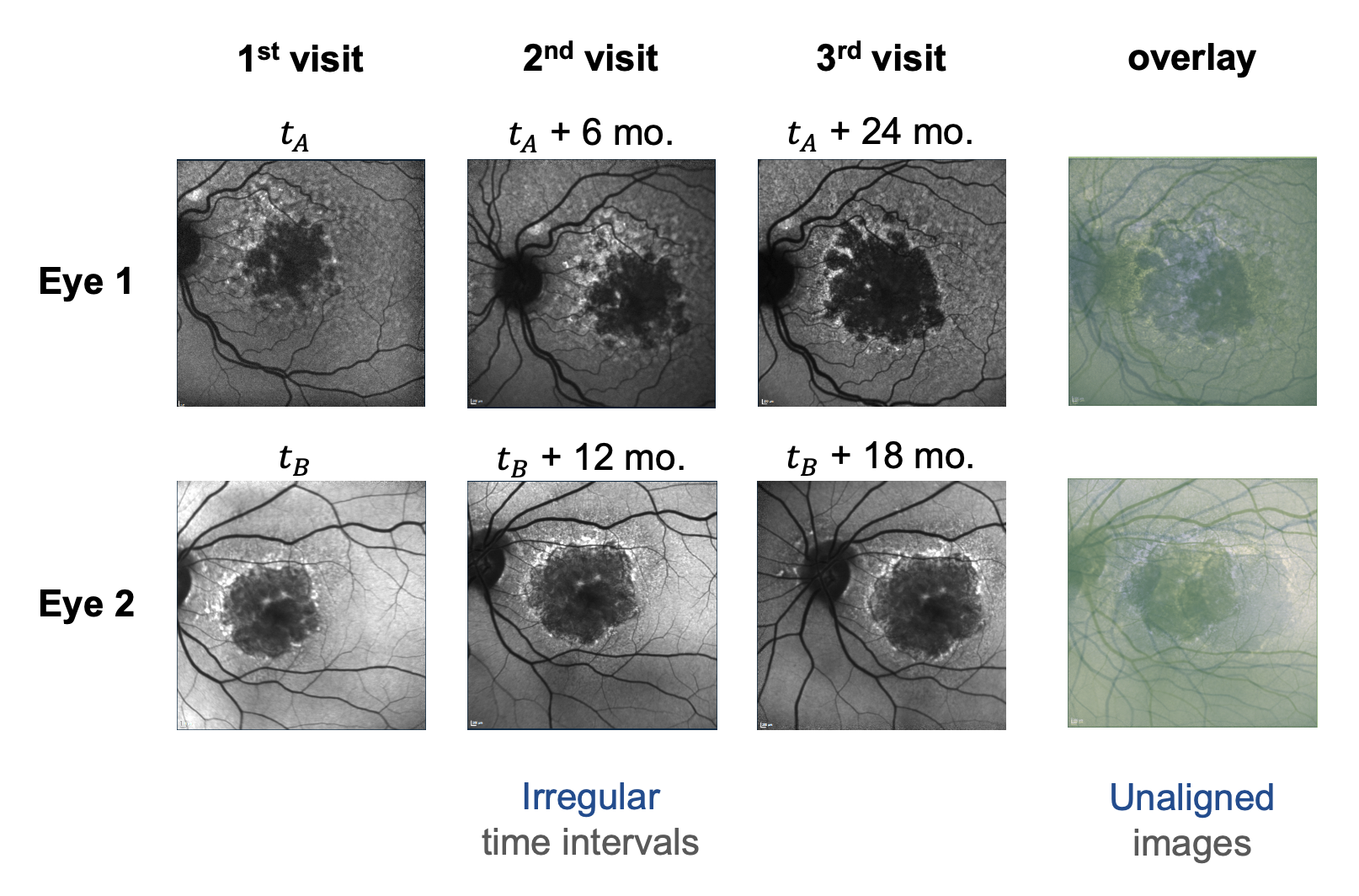}
    \caption{Temporal sparsity, sampling irregularity and spatial misalignment in longitudinal images.}
    \label{fig:data}
\end{figure*}

\fancynumber{2} \textbf{Temporal Sparsity} is especially common in longitudinal images in healthcare, as images are usually acquired at separate visits of the patient, where the time gap can be several months or years. In contrast, a relatively well-studied adjacent field is video data~\cite{bar2024lumiere, ho2022imagen, hoppe2022diffusion}, where the frame rate can easily be 60 Hz or higher. This renders our data of interest easily $10^8$ times sparser compared to the better studied video data.

\fancynumber{3} \textbf{Sampling Irregularity} is also ubiquitous in clinical practice, both \textit{within} and \textit{among} longitudinal image series. \textit{Within-series} irregularity means that the visits are not necessarily evenly distributed for the same patient over time. \textit{Among-series} irregularity means that different patients do not follow the same readmission schedule either --- \textit{in terms of both time intervals and number of visits}. Times for visits can significantly vary based on doctors' evaluation of the condition, the availability of doctors and/or imaging facilities, and the patient's own preferences, among others. This feature defies the assumptions of most methods that require regular sampling or common sampling~\cite{MIOFlow, survey_video_diffusion}.

\fancynumber{4} \textbf{Spatial Misalignment} is often seen in longitudinal medical images too. Indeed, it is almost impossible to enforce pixel-perfect alignment of images acquired at different visits. Luckily, this problem can be addressed by image registration without any compounding effect with the temporal sparsity or sampling irregularity issues. See Appendix~\ref{supp:image_registration} for an illustration of image registration.

Temporal sparsity, sampling irregularity, and spatial misalignment are illustrated in Figure~\ref{fig:data}. These properties and challenges listed above lead to \textbf{a fairly unique area of research} that is \textbf{largely underexplored but highly interesting} to healthcare professionals.

Consider retinal imaging as an example. Most existing approaches to estimate disease progression in retinal imaging
data do not operate in the image space, but rather in a vector space of a few clinical features extracted from the images. Examples of these derived statistics include the area of geographic atrophy lesions~\cite{GA_prog_pred1}, the number of lesions~\cite{GA_prog_pred2}, the lesion perimeter~\cite{GA_prog_pred1}, its prior observed growth rate~\cite{GA_prog_pred2}, the presence and pattern of hyperfluorescence around the border of a lesion~\cite{Shen2021LocalProgression}. Although these approaches have been effective, they compress the rich context in the images to just a few metrics, and the output is an oversimplified representation of the disease states. This simplification overlooks the nuanced variations and complexities that are discarded during the feature extraction process and limits the interpretability of the output to a few preselected scalar-valued features.

In contrast, our proposed ImageFlowNet capitalizes on the extensive information available in the image to provide a nuanced representation of future conditions and also addresses the limitations of traditional metrics-based methodologies by offering a more dynamic and detailed visualization of disease progression. This method gives healthcare professionals an intuitive understanding of the expected progression of the disease and allows them to provide patients with a visual forecast that goes beyond mere numerical data.

We hope that our method can establish a new standard in the discipline and potentially transform clinical practices in areas including but not limited to ophthalmology or neurology, with the help of the latest imaging and measurement techniques~\cite{hussain2022modern, photoacustic_tomography} as well as computational tools for disease diagnosis~\cite{SkinCancerCls, BreastCancerDet, SkinCancerDet}, risk prediction~\cite{RiskPred1, RiskPred2}, uncertainty quantification~\cite{uncertainty_quantification1, uncertainty_quantification2}, planning~\cite{disease_planning1, RL1, Optimization1, RL2}, and patient care~\cite{Pred_treatment_Longitudinal, xie2022deepvs, xie2021vitalhub}.

\section{Additional Background on Why Time-Awareness is Important}

Solving our problem outlined in \textcolor{YaleBlue}{Section~\ref{sec:prelim}} with deep learning requires designing and optimizing a model $\mathcal{F}: (\mathbb{R}^{H \times W \times C}, \mathbb{R}, \mathbb{R}) \rightarrow \mathbb{R}^{H \times W \times C}$, such that $\widehat{x_j} = \mathcal{F}(x_i, t_i, t_j)$ and $\widehat{x_j} \approx x_j$.

In most existing image-to-image tasks, the mapping between each pair of input $x_i$ and output $x_j$ obeys the same transformation rules, and hence their models are designed to be time-agnostic. For example, in denoising~\cite{ImageDenoising, DeepImDenoising}, $x_j$ is the noise-free version of $x_i$; in super-resolution~\cite{Image_superresolution, SRCNN, superresolution1, superresolution2}, $x_j$ is higher in resolution than $x_i$ by a fixed factor; in reconstruction~\cite{NonsmoothNonconvexRecon, reconstruction1, fastMRI, reconstruction2, ding2023learned, bian2022learnable}, $x_j$ is the transformed version of $x_i$ through a fixed set of rules guided by physics; in contrast mapping~\cite{Contrast_mapping1, Contrast_mapping2, Contrast_mapping3, Contrast_mapping4, Contrast_mapping5}, $x_j$ represents the effect of staining or contrast agents when applied to $x_i$; and in segmentation~\cite{Segmentation1, Segmentation2, Segmentation3, Segmentation4, Segmentation5, Segmentation6}, $x_j$ returns a label map describing the anatomical or functional segments in $x_i$. For these purposes, time-agnostic models, such as UNet or most diffusion models\footnote{While diffusion models have modules that can encode time, many variants are used in a time-agnostic manner for tasks like denoising or super-resolution, where ``time'' is no different from ``iteration''.} remain competitive.

However, in our scenario, the output image is a function of both the input image and time. Given the same input image $x_i$, it will not end up at the same output image if the time interval changes. An image showing a disease 2 years after onset may look very different compared to 2 days after onset. In such cases, attempting to solve this problem using a model without time-modeling capabilities would be fundamentally ill-posed. In short, the spatial-temporal problem requires a spatial-temporal solution, which inspired our development of ImageFlowNet.

\clearpage
\newpage

\section{Image Registration}
\label{supp:image_registration}

\subsection{Retinal Images}
For all images, we extracted descriptive keypoints with SuperRetina~\cite{SuperRetina}, a high-quality keypoint detector trained on retinal images. Then we identified the keypoint correspondences for each image pair in each longitudinal series with a $k$-nearest-neighbor matcher and considered any image pair that has at least 15 keypoint correspondences a successful match. Next, we selected the image that produced the most successful matches as the ``anchor image''. Finally, we aligned all images in the longitudinal series towards the ``anchor image'' using perspective transformation so that the degree of freedom is constrained to the adjustment of camera angle or position. As a post-processing step, for each longitudinal series, we cropped all images with the biggest common foreground square so that no image contained any background pixel outside the retina region.

The image registration process for a pair of images from the same longitudinal series is illustrated in Figure~\ref{fig:registration}. It can be seen that all veins are aligned in the resulting images while atrophy borders are not. This is expected from perspective transformation and is exactly desirable for our task. 

\begin{figure*}[!thb]
    \centering
    \includegraphics[width=\textwidth]{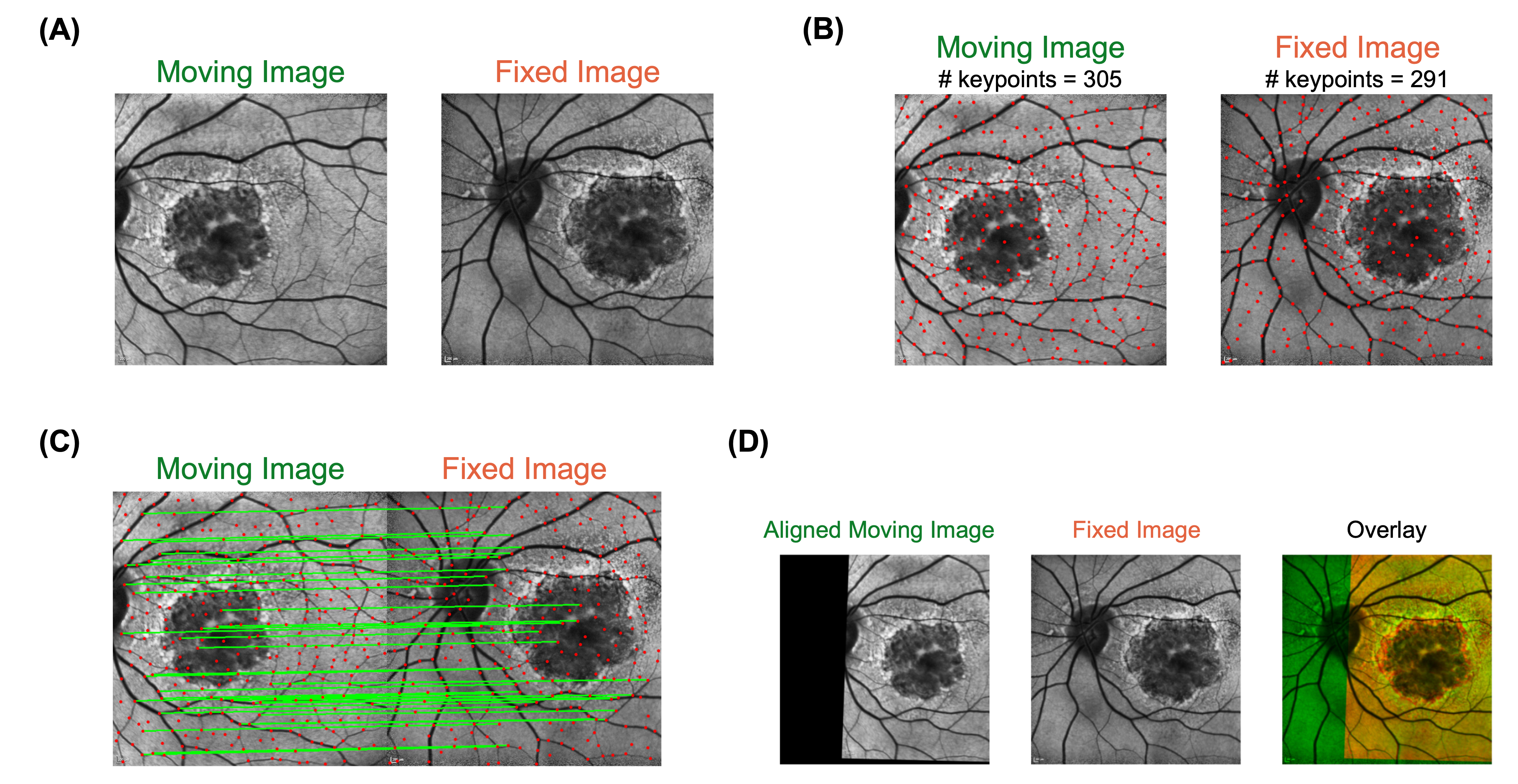}
    \caption{Our image registration pipeline. (A) Moving and fixed images come from the same eye at different time points. (B) SuperRetina is used to detect consistent and descriptive keypoints. (C) Keypoints are matched by descriptor similarity and filtered by distance heuristics. (D) The moving image is aligned under the constraint of a perspective transformation. }
    \label{fig:registration}
\end{figure*}

\subsection{Brain Multiple Sclerosis Images}
These images were already registered. No additional work was done.

\subsection{Brain Glioblastoma Images}
We used the scans in the ``DeepBraTumIA'' folders, which were registered to a common atlas, but the registration did not adequately align the scans in each longitudinal series. We used the Python tool from ANTS~\cite{ANTS} to perform \textit{Affine} followed by \textit{Diffeomorphic} registration with [4, 2, 1] iterations to align each scan towards the first scan in series.

\clearpage
\section{Implementation Details}
\label{supp:implementation_details}

\paragraph{Architectures} The proposed ImageFlowNet combines UNet and Neural ODEs. The UNet model follows the time-conditional UNet implementation in Guided Diffusion~\cite{GuidedDiffusion}. Neural ODEs are implemented with torchdiffeq~\cite{torchdiffeq}.

\paragraph{Data Augmentation} We used the albumentations package~\cite{Albumentations} to perform flipping, shifting, scaling, rotation, random brightness, and random contrast. We also make the UNet training a denoising process by adding random Gaussian noise to the input.

\paragraph{Hyperparamters and training details} All experiments were performed on a SLURM server, where each job was allocated either an NVIDIA A100 GPU, an NVIDIA A5000 GPU, or an NVIDIA RTX 3090 GPU. All jobs can be completed within 2-5 days on a single GPU with 8 CPU cores. T-Diffusion usually takes the longest to train. ImageFlowNet$_\textrm{SDE}$ variant may require a 40-GB GPU (sometime that will still hit an OOM error if running too many function evaluations in the SDE) while all other methods can be trained on a 20-GB GPU. Experiments shared the same set of hyperparameters: learning rate = 0.0001, batch size = 64, number of epochs = 120. Adam with decoupled weight decay~(AdamW)~\cite{AdamW} optimizer was used, along with a cosine annealing learning rate scheduler with linear warmup. We used an exponential moving average~(EMA) with decay rate of 0.9 on the ImageFlowNet models.

To accommodate the GPU VRAM limits, we used gradient aggregation to trade efficiency for space while achieving the desired effective batch size --- we used an actual batch size of 1, scaled the loss by $\frac{1}{64}$, and updated the weights every 64 batches.

Training of the segmentation networks are described in the next section (Evaluation Metrics).

\clearpage
\newpage
\section{Evaluation Metrics}
\label{supp:metrics}

The evaluation metrics cover image similarity, residual magnitude, and atrophy similarity.

\paragraph{Image similarity}
We measure the image similarity between the real future image $x_j$ and the predicted future image $\widehat{x_j}$ using peak signal-to-noise ratio~(PSNR) and structural similarity index~(SSIM). These two metrics are widely used in image-to-image tasks such as super-resolution, denoising, inpainting, etc.

PSNR is a normalized version of the mean squared error between two images that takes into account the dynamic range of the image data. The formula is given by Eqn~\eqref{eqn:psnr}.

\begin{align}
    \label{eqn:psnr}
    \mathrm{PSNR}(x_a, x_b) &= 10 \log_{10} \left( \frac{R}{\mathrm{MSE}(x_a, x_b)} \right) \textrm{, where} \\
    \nonumber R &\textrm{ is the common dynamic range of the images}\\
    \nonumber \mathrm{MSE}(x_a, x_b) &= \frac{1}{H \times W} \sum_{h \in H, w \in W} ||x^{(h,w)}_a - x^{(h,w)}_b||^2
\end{align}

SSIM measures the similarity between two images by describing the perceived change in structural information. The formula is given by Eqn~\eqref{eqn:ssim}. We used the implementation in Scikit-image~\cite{scikit-image}.

\begin{align}
    \label{eqn:ssim}\mathrm{SSIM}(x_a, x_b) &= \frac{(2 \mu_{x_a} \mu_{x_b} + c_1)(2 \sigma_{x_a x_b} + c_2)}{(\mu^2_{x_a} + \mu^2_{x_b} + c_1)(\sigma^2_{x_a} + \sigma^2_{x_b} + c_2)} \textrm{, where} \\
    \nonumber \mu_{x_a} &\textrm{ is the pixel sample mean of } x_a \\
    \nonumber \mu_{x_b} &\textrm{ is the pixel sample mean of } x_b \\
    \nonumber \sigma^2_{x_b} &\textrm{ is the variance of } x_b \\
    \nonumber \sigma^2_{x_b} &\textrm{ is the variance of } x_b \\
    \nonumber \sigma_{x_b x_b} &\textrm{ is the covariance of } x_a \textrm{ and } x_b \\
    \nonumber c_1 &= (0.01 R)^2, c_2 = (0.03 R)^2\\
    \nonumber R &\textrm{ is the common dynamic range of the images}
\end{align}

\paragraph{Residual magnitude}
We evaluated the magnitude of the residual maps $\widehat{x_j} - x_j$ using the mean average error (MAE) and the mean squared error (MSE).

\paragraph{Atrophy similarity}
We also want to emphasize the precise representation of the atrophy region. To this end, the simplest metric is the dice similarity coefficient~(DSC) and Hausdorff distance~(HD) of the binarized atrophy regions. DSC and HD between two binary masks $X$ and $Y$ are given by Eqn~\eqref{eqn:dsc} and Eqn~\eqref{eqn:hd}, respectively. For HD, we used the implementation in Scikit-image~\cite{scikit-image}.

\begin{align}
    \label{eqn:dsc}\mathrm{DSC}(X, Y) &= \frac{|X \cap Y|}{|X| + |Y|}
\end{align}

\begin{align}
    \label{eqn:hd}\mathrm{HD}(X, Y) &= \max \left\{ \sup_{x \in X} d(x, Y), \sup_{y \in Y} d(X, y) \right\}
\end{align}

To perform atrophy segmentation, we separately trained three auxiliary image segmentation network on all images, one for each dataset. All retinal images have their atrophy regions labeled by ophthalmologists. All brain images have associated segmentation maps from the dataset providers. These segmentators that we trained have an nn-UNet~\cite{nnUNet} architecture and were trained with an AdamW~\cite{AdamW} optimizer at an initial learning rate of 0.001 for 120 epochs. With these networks, we can segment the atrophy regions in both the real future image $x_j$ and the predicted future image $\widehat{x_j}$. DSC and HD can be computed on the segmentation masks between each pair of interest.

\vspace{12pt}
\end{appendices}


\end{document}